\pdfoutput=1
\documentclass[3p,authoryear]{elsarticle}

\usepackage{amsmath,amssymb,amsfonts,amsthm}
\usepackage{graphicx}
\usepackage{booktabs}
\usepackage{siunitx}
\usepackage{algorithm}
\usepackage{algorithmic}
\usepackage{multirow}
\usepackage{xcolor}
\usepackage{colortbl}
\usepackage{bm}
\usepackage{lineno}
\usepackage{enumitem}
\usepackage{tcolorbox}
\usepackage{tikz}
\usetikzlibrary{positioning,shapes.geometric,arrows.meta}
\usepackage{hyperref}
\usepackage{cleveref}
\crefformat{section}{#2#1#3}
\Crefformat{section}{#2#1#3}

\providecommand{\doi}[1]{\href{https://doi.org/#1}{\texttt{doi:}#1}}

\newtheorem{assumption}{Assumption}
\newtheorem{theorem}{Theorem}
\newtheorem{lemma}{Lemma}
\newtheorem{proposition}{Proposition}
\newtheorem{corollary}{Corollary}
\newtheorem{remark}{Remark}

\crefname{assumption}{Assumption}{Assumptions}
\Crefname{assumption}{Assumption}{Assumptions}
\crefname{theorem}{Theorem}{Theorems}
\Crefname{theorem}{Theorem}{Theorems}
\crefname{lemma}{Lemma}{Lemmas}
\Crefname{lemma}{Lemma}{Lemmas}
\crefname{proposition}{Proposition}{Propositions}
\Crefname{proposition}{Proposition}{Propositions}
\crefname{corollary}{Corollary}{Corollaries}
\Crefname{corollary}{Corollary}{Corollaries}
\crefname{remark}{Remark}{Remarks}
\Crefname{remark}{Remark}{Remarks}
\crefname{definition}{Definition}{Definitions}
\Crefname{definition}{Definition}{Definitions}

\newcommand{\R}{\mathbb{R}}
\newcommand{\xvec}{\mathbf{x}}
\newcommand{\zvec}{\mathbf{z}}
\newcommand{\uvec}{\mathbf{u}}
\newcommand{\vvec}{\mathbf{v}}
\newcommand{\evec}{\mathbf{e}}
\newcommand{\wvec}{\mathbf{w}}
\newcommand{\emax}{\bar{\mathbf{e}}}
\newcommand{\wmax}{\bar{\mathbf{w}}}

\journal{Advances in Space Research}

\begin{document}

\begin{frontmatter}

\title{Horizon-Dependent Tube MPC for Spacecraft Rendezvous on Elliptical Orbits with Conditional Robust Constraint Satisfaction}

\author[inst1]{Omer Burak Iskender\corref{cor1}}
\cortext[cor1]{Corresponding author}
\ead{iske0001@e.ntu.edu.sg}

\address[inst1]{School of Electrical and Electronic Engineering, Nanyang Technological University, Singapore}

\begin{abstract}
Spacecraft rendezvous on elliptical orbits must hold a safety corridor under navigation noise, unmodelled perturbations, and thrust errors driven by propellant mass uncertainty. This paper develops a tube-based model predictive controller for the Yamanaka--Ankersen linear time-varying dynamics, carrying constraint-tightening tube methods from circular to elliptical orbits. A horizon-dependent, element-wise error bound propagates the actual closed-loop matrices along the prediction window, so early prediction steps keep nearly the full corridor that a constant-width tube would surrender. A multiplicative-to-additive conversion folds mass and thrust uncertainty into the same tightening recursion, and a Perron--Frobenius spectral-radius condition supplies a computable certificate that the tightening converges, with an explicit input-to-state stability gain. In a paired Monte Carlo campaign on a Mars sample-return orbit, the tube controller cuts mean corridor violations by more than an order of magnitude against a nominal predictive baseline at comparable fuel cost. A nonlinear truth-model test with oblateness perturbations well beyond the assumed disturbance budget leaves the corridor unviolated, and when the design envelope is exceeded the tightened problem becomes infeasible and the controller reverts to a saturated linear fallback, making the loss of guarantee explicit rather than silent. The construction requires only element-wise arithmetic and an open-source quadratic-programming solver.
\end{abstract}

\begin{keyword}
tube MPC \sep spacecraft rendezvous \sep elliptical orbit \sep robust control \sep LTV systems
\end{keyword}

\end{frontmatter}

\begin{center}
\small\itshape
Accepted for publication in Advances in Space Research.
\end{center}

\section*{Nomenclature}
\begin{tabbing}
\hspace{2.8cm}\= \kill
$a$ \> = semi-major axis, m \\
$A_k$ \> = discrete YA state matrix at step $k$, $\R^{6\times 6}$ \\
$\bar{A}$ \> = element-wise worst-case matrix, $\max_k |A_{\mathrm{cl},k}|$ \\
$A_{\mathrm{cl},k}$ \> = closed-loop matrix, $A_k + B_k K$ \\
$B_k$ \> = discrete input matrix at step $k$, $\R^{6\times 3}$ \\
$e$ \> = orbital eccentricity \\
$\evec_k$ \> = tracking error, $\xvec_k - \zvec_k$ \\
$\emax_j$ \> = element-wise error bound at horizon step $j$ \\
$\gamma_{\mathrm{ISS}}$ \> = input-to-state stability gain \\
$K$ \> = tube feedback gain, $\R^{3\times 6}$ \\
$N$ \> = MPC prediction horizon (steps) \\
$n$ \> = mean motion, rad/s \\
$\nu$ \> = true anomaly, rad \\
$P$ \> = terminal cost weight matrix \\
$Q,\;R$ \> = state and input running cost matrices \\
$\rho(\cdot)$ \> = spectral radius \\
$\alpha,\;\alpha_N$ \> = terminal level set and its horizon-deflated value \\
$\lambda_\star$ \> = worst-case contraction-loss ratio in the $P$-metric \\
$\varepsilon$ \> = DARE regularisation scalar (\Cref{eq:DARE}) \\
$\varepsilon_*$ \> = Young-inequality weight in \Cref{eq:alpha} \\
$\mathcal{X}_f$ \> = ellipsoidal terminal set \\
$\nu_\star$ \> = worst-case orbit phase for the terminal-set design \\
$T_s$ \> = sampling period, s \\
$\uvec_k$ \> = applied control input (velocity increment), m/s \\
$\vvec_k$ \> = nominal planned input (QP variable), m/s \\
$\wvec_k$ \> = additive bounded disturbance \\
$\wmax$ \> = process-noise disturbance bound vector \\
$\wmax_j$ \> = per-step total disturbance bound (process noise $+$ mass) \\
$\delta_{m,\max}$ \> = maximum mass/thrust uncertainty fraction \\
$\xvec_k$ \> = relative state, $[r_x,\;r_y,\;r_z,\;v_x,\;v_y,\;v_z]^\top$ \\
$\xvec_s$ \> = setpoint state \\
$\zvec_k$ \> = nominal state trajectory
\end{tabbing}

\section*{Notation}
\label{sec:notation}

For a vector $v\in\R^n$ and a matrix $M\in\R^{m\times n}$, $|v|$ and $|M|$ denote element-wise absolute values; inequalities between vectors or matrices of equal dimension are read element-wise. $\rho(M)=\max_i|\lambda_i(M)|$. The row-wise triangle inequality $|Mv|\le|M|\,|v|$ and, for non-negative $A\le B$, $\rho(A)\le\rho(B)$~\cite{Horn_Johnson}. ``Element-wise box'' is the hyper-rectangle $\{x:|x|\le\bar x\}$.

\section{Introduction}
\label{sec:intro}

Active debris removal~\cite{ClearSpace2021,Astroscale2022}, in-orbit servicing, and planetary sample return require autonomous proximity operations on orbits whose eccentricity is set by insertion accuracy and mission budget: LEO values of $0.001$--$0.01$, GTO up to $0.73$, Mars-capture ellipses of $0.05$--$0.3$ are all routine~\cite{Sullivan2017,Fehse2003}. The Hill--Clohessy--Wiltshire (HCW) model~\cite{HCW1960} underpinning most rendezvous studies assumes a constant orbital rate; position-propagation error then grows quadratically with $e$ and linearly with horizon. At $e=0.2$ a 30-step HCW prediction departs from Keplerian propagation by kilometres (\Cref{fig:ya_vs_cw}); a controller using HCW must inflate robustness margins or accept violations.

\paragraph{A. Rendezvous on elliptical orbits.}
Few constrained RPO studies move beyond CW. Inalhan et~al.~\cite{Inalhan2002} demonstrated fuel-optimal planning on the YA STM~\cite{YA_original}; Hartley et~al.~\cite{Hartley2015b} solved the LTV MPC in real time on FPGA on the identical MSRE orbit; neither provides tube-based robustness. In ASR, Lim et~al.~\cite{Lim2018} used heuristic constraint backoffs without invariant-set tightening. The robustness mechanism on eccentric orbits is therefore either absent or heuristic.

\paragraph{B. Tube-based MPC for spacecraft.}
Tube MPC~\cite{Langson2004,Mayne2005} decomposes the control into a nominal QP input and feedback $K(\xvec_k-\zvec_k)$ that confines tracking error to a bounded set; for LTI systems the tube width is the mRPI set~\cite{Rakovic2005}. Surveys~\cite{Eren2017,Mayne2014} identify robust constraint satisfaction as an open problem in MPC for aerospace~\cite{Weiss2015}.

Within ASR, Dong et~al.~\cite{Dong2020} published the first tube-based MPC for rendezvous, with constant-width mRPI on CW dynamics and output feedback; a follow-up~\cite{Dong2024} added an online disturbance observer that narrows the tube under mild perturbations but stays on CW with constant-width tubes. Zhu et~al.~\cite{Zhu2018} use worst-case reachable-set backoffs without the tube decomposition. Kang et~al.~\cite{Kang2025} tighten using the closed-loop rather than open-loop reachable set, reducing conservatism but on CW.

Outside ASR, Specht and Mooij~\cite{Specht2023} designed multi-phase tube MPC for Envisat capture with constant polytopic tubes, Oestreich and Linares~\cite{Oestreich2023} proposed adaptive ellipsoidal tubes updated via per-step SDP, and Mammarella et~al.~\cite{Mammarella2017,Mammarella2020} validated box-tube MPC on an air-bearing testbed, all under LTI dynamics. Bokor et~al.~\cite{Bokor2025} addressed robust MPC for rendezvous under sector-bounded nonlinearities; the formulation handles aerodynamic and gravity-gradient nonlinearities but inherits CW dynamics. On the general LTV front, Yu and Cannon~\cite{YuCannon2016} established tube MPC with time-varying constraint tightening, and K\"ohler et~al.~\cite{Kohler2019} extended the ideas to tracking. Closest to the present work, Bumroongsri~\cite{Bumroongsri2015} formulated tube-based MPC for LTV systems with a horizon-dependent invariant set for chemical-process applications without addressing orbital mechanics, mass uncertainty, or eccentric-orbit specialisation. The present manuscript builds on that horizon-dependent tightening and specialises it to the YA dynamics.

\paragraph{C. Safety corridors and keep-out zones.}
Richards et~al.~\cite{Richards2002} formulated mixed-integer LPs for obstacle avoidance; Gavil\'an et~al.~\cite{Gavilan2012} added chance constraints. In ASR, Wang et~al.~\cite{WangASR2023} iteratively convexify non-convex keep-out constraints around the current trajectory, enforced on the nominal path without a disturbance margin. Iskender et~al.~\cite{Iskender2019,Iskender2020} validated MPC docking with dual-quaternion kinematics in simulation and hardware. All these methods rely on stochastic margins or empirical backoffs; none provides a deterministic worst-case bound on constraint violation.

\paragraph{D. Uncertainty modelling in relative motion.}
Uncertainty in RPO arises from navigation errors, propellant-driven mass changes, thruster execution errors, and unmodelled perturbations ($J_2$, third-body). In ASR, Zheng and Luo~\cite{Zheng2024} performed closed-loop uncertainty analysis for non-cooperative rendezvous, propagating linearised covariance to generate robust trajectory corridors, a stochastic rather than set-membership approach. The multiplicative-to-additive conversion introduced here absorbs mass uncertainty $\delta_m$ into a per-step additive bound $\wmax_j = \wmax + \delta_{m,\max}|B_{k+j}|u_{\max}$, capturing parametric uncertainty within the box-tightening framework without online covariance propagation or SDP-based tube updates.

\paragraph{E. Gap and contribution statement.}
Horizon-dependent tightening for LTV systems exists generically~\cite{YuCannon2016,Bumroongsri2015,Kohler2019}; this paper specialises it to spacecraft rendezvous on elliptical orbits with four additions that do not appear in the cited LTV tube MPC literature:
\begin{enumerate}[label=(C\arabic*),leftmargin=2.5em,nosep]
  \item \textbf{YA specialisation with an orbit-specific spectral-radius certificate.} The Perron--Frobenius condition $\rho(\bar{A})<1$, where $\bar{A}=\max_k|A_k+B_kK|$ is taken element-wise over one orbital period, is verified analytically and numerically for the Yamanaka--Ankersen STM, giving a computable eccentricity threshold $e_{\max}$ below which the box-tube tightening converges. This converts a generic LTV requirement into a closed-form rendezvous certificate.
  \item \textbf{Mass-aware multiplicative-to-additive lifting.} A per-step bound $\bar{w}_j = \bar{w} + \delta_{m,\max}|B_{k+j}|u_{\max}$ folds parametric thrust/mass uncertainty into the same additive tightening recursion that handles process noise, eliminating the need for a separate parametric tube. The per-step structure follows the time-varying $B_k$ of the YA model.
  \item \textbf{Ellipsoidal terminal invariant set for the LTV closed loop.} A single ellipsoid $\mathcal{X}_f=\{x:(x-x_s)^\top P(x-x_s)\le\alpha\}$ computed offline from the DARE at the worst-case orbit phase closes, in theory, the recursive-feasibility gap that constant-width tubes leave open. This contribution is derived and its level set computed in closed form, but the Monte Carlo campaign reported here runs without it (\Cref{rem:terminal_scope}); enforcing it exactly requires a second-order-cone solver rather than the QP solver used throughout.
  \item \textbf{Empirical conservatism map and direct comparison against a classical constant-width tube on the identical scenario}, demonstrating that the horizon-dependent recursion delivers a $\sim 3.5\times$ reduction in violations and a $\sim 3.4\times$ reduction in fallback events at the design envelope boundary (\Cref{sec:conservatism_compare_mc}).
\end{enumerate}
The construction uses element-wise arithmetic plus one quadratic terminal constraint solvable by the open-source OSQP solver; no polytope vertex enumeration, no per-step SDP, and no commercial software are required. Augmented-dynamics extensions (fuel slosh, modal flexibility) are pursued in companion formulations~\cite{companion_slosh,companion_flex}. All Monte~Carlo results are bit-reproducible from the seed-deterministic Python implementation released with the paper.

A preliminary version of the horizon-dependent tightening recursion (C2) and the Perron--Frobenius certificate (C1) was presented at the 77th International Astronautical Congress~\cite{Iskender2026EllipticalTube}, where they were applied to a phased Hartley-style approach and a geostationary co-location case study. The present paper extends that work in four directions: the ellipsoidal terminal set and the recursive-feasibility and input-to-state-stability results of \Cref{sec:tube} (C3), which the conference version does not carry, derived here but not exercised in the reported campaign (\Cref{rem:terminal_scope}); the paired comparison against a constant-width tube on an identical scenario (C4); a nonlinear truth-model test with $J_2$ oblateness driven well past the assumed disturbance budget; and the cross-body envelope study spanning Mars, low Earth, geostationary transfer, and lunar orbits.

\section{Problem Formulation}
\label{sec:problem}

The scenario and formulation match Hartley et~al.~\cite{Hartley2012}: identical MSRE orbit ($a=4643$~km, $e=0.204$), YA relative dynamics, impulsive delta-$v$, and receding-horizon LTV MPC. The contribution is to replace that controller with a horizon-dependent tube MPC handling mass uncertainty and unmodelled disturbances without constraint violation. This section records the dynamics, assumptions, and constraints used by both.

\subsection{Yamanaka--Ankersen Relative Dynamics}

Consider two spacecraft orbiting a central body: a \emph{chief} (target) on a known Keplerian ellipse, and a \emph{deputy} (chaser) whose position and velocity are measured relative to the chief in a local-vertical--local-horizontal (LVLH) frame. The $x$-axis points along the chief's velocity direction (V-bar), $z$ points radially outward (R-bar), and $y$ completes the right-handed triad (cross-track, H-bar). The relative state is
\begin{equation}\label{eq:state}
  \xvec = \begin{bmatrix} r_x & r_y & r_z & v_x & v_y & v_z \end{bmatrix}^\top \in \R^6,
\end{equation}
with positions in metres and velocities in metres per second.

The relative motion is governed by the linearised Tschauner--Hempel equations~\cite{TH1965}, whose closed-form solution is the Yamanaka--Ankersen state transition matrix (STM) $\Phi(\nu,\nu_0)\in\R^{6\times 6}$~\cite{YA_original}. With zero-order-hold sampling at period $T_s$, the discrete-time model is
\begin{equation}\label{eq:dynamics}
  \xvec_{k+1} = A_k\,\xvec_k + B_k\,\uvec_k + \wvec_k,
  \qquad
  A_k = \Phi(\nu_{k+1},\nu_k),\quad
  B_k = A_k\begin{bmatrix} \mathbf{0}_{3\times 3}\\ I_3 \end{bmatrix},
\end{equation}
where $\uvec_k\in\R^3$ is the impulsive delta-$v$ applied at the start of the interval and $\wvec_k\in\R^6$ is the additive disturbance. Both $A_k$ and $B_k$ depend on the orbit phase $\nu_k$ and trace a periodic cycle over one orbit; as $e\to 0$ the YA STM reduces continuously to the HCW matrix. The MSRE-orbit propagation gap between the YA and HCW models, and the residual against the nonlinear two-body truth, are quantified in \Cref{fig:ya_vs_cw}.

\begin{figure}[htbp]
  \centering
  \includegraphics[width=0.7\textwidth]{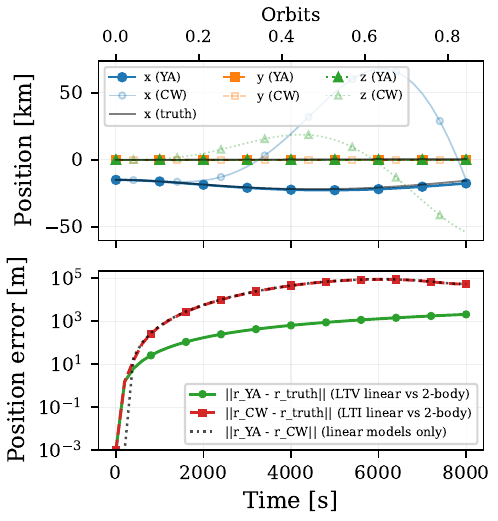}
  \caption{Free-drift propagation comparison on the MSRE orbit ($e=0.204$). The CW model (dashed) accumulates kilometre-scale position errors relative to the exact YA model (solid) over a 30-step prediction horizon. Grey shading indicates the absolute error magnitude.}
  \label{fig:ya_vs_cw}
\end{figure}

\subsection{Assumptions}

The framework rests on four standing hypotheses, stated here so that each later guarantee can name the ones it uses.

\begin{assumption}[Keplerian Chief Orbit and Linearisation Validity]\label{ass:dynamics}
The chief spacecraft follows a Keplerian elliptical orbit with known classical elements $(a,e,i,\Omega,\omega)$ and eccentricity $0<e<1$. The relative separation between chaser and chief is small enough relative to the orbit radius that the linearised Tschauner--Hempel equations capture the relative dynamics within the disturbance budget of \Cref{ass:bounded_w}.
\end{assumption}

\begin{remark}[Scope of \Cref{ass:dynamics}]\label{rem:ass1_scope}
\Cref{ass:dynamics} is a small-separation linearisation. It is well satisfied during the corridor approach, station-keeping, and waypoint-transfer phases on which the present study evaluates the controller (separations $\le 15$~km, well below $1\%$ of the MSRE orbit radius). It is \emph{not} intended to cover far-rendezvous phasing manoeuvres at thousands of kilometres separation, nor terminal contact dynamics inside the docking envelope where contact forces dominate. The cross-body validation of \Cref{sec:cross_body} respects the same scope: corridor sizes are chosen to keep the second-order nonlinearity below the medium-tier additive bound $\bar{w}$, as quantified by the two-body validation in \Cref{sec:j2_truth} (residual $\sim 2.75$~m at 15~km offset).
\end{remark}

\begin{assumption}[Bounded Disturbance]\label{ass:bounded_w}
The disturbance satisfies $|\wvec_k|\leq\wmax$ element-wise for all $k\geq0$, where $\wmax\in\R^6_{\geq0}$ is known. This bound encompasses residual $J_2$ differential acceleration, solar radiation pressure, third-body perturbations, navigation-filter output noise, discretisation error, and any other bounded exogenous effect.
\end{assumption}

\begin{assumption}[State Feedback]\label{ass:state_feedback}
The controller has access to an estimate $\hat{\xvec}_k$ of the relative state $\xvec_k$ at each sampling instant. The estimation error $\xvec_k-\hat{\xvec}_k$ is treated as one of the bounded exogenous effects absorbed into $\wmax$ of \Cref{ass:bounded_w}; consequently, residual navigation-filter noise, sensor bias, and discretisation error are all carried inside the tightening recursion. The construction of the estimator itself (typical implementations couple a Kalman filter on LIDAR/camera relative-pose measurements~\cite{Iskender2019,Iskender2020}) is outside the scope of this paper.
\end{assumption}

\begin{assumption}[Controllability]\label{ass:controllable}
The pair $(A_k,B_k)$ is controllable for every $k$. This is guaranteed by three-axis thrust authority and the non-singularity of the YA STM for $e<1$.
\end{assumption}

\subsection{Constraint Formulation}

The actuator imposes element-wise bounds on the velocity increment:
\begin{equation}\label{eq:uconstraint}
  u_{\min}\leq\uvec_k\leq u_{\max},\qquad
  u_{\min}=-u_{\mathrm{a}}\mathbf{1}_3,\quad u_{\max}=+u_{\mathrm{a}}\mathbf{1}_3,
\end{equation}
with per-axis authority $u_{\mathrm{a}} = \SI{5}{\metre\per\second}$
(\Cref{tab:tiers}).
The state is confined to a safety corridor:
\begin{equation}\label{eq:xconstraint}
  x_{\min}\leq\xvec_k\leq x_{\max},
\end{equation}
with the specific values of $x_{\min}$ and $x_{\max}$ depending on the scenario (approach corridor, station-keeping box, or waypoint transfer) as detailed in \Cref{sec:setup}.

\begin{remark}[Disturbance budget]\label{rem:disturbance}
The dominant unmodelled acceleration on the MSRE orbit is the differential $J_2$ effect ($J_2=1.96{\times}10^{-3}$ for Mars), producing $\delta a_{J_2}\approx10^{-6}$~m/s$^2$ at $\Delta r=10$~km and yielding per-step bounds $\Delta v\approx2{\times}10^{-4}$~m/s, $\Delta r\approx0.02$~m over $T_s=200$~s. Navigation noise adds ${\sim}1$~m position and ${\sim}0.01$~m/s velocity per sample. The ``medium'' disturbance tier (\Cref{tab:tiers}) encompasses this budget; the ``extreme'' tier stress-tests beyond it.
\end{remark}

\section{Tube-Based MPC}
\label{sec:tube}

\subsection{Nominal MPC Inner Loop}

The tube MPC solves the same QP structure as nominal MPC~\eqref{eq:nomMPC}, with one critical modification: the constraints are \emph{tightened} to leave room for the worst-case tracking error. The QP decision variables are the nominal states $\zvec_{0:N}$ and nominal inputs $\vvec_{0:N-1}$, and the nominal dynamics propagate disturbance-free:
\begin{equation}\label{eq:nom_dyn}
  \zvec_{j+1} = A_{k+j}\,\zvec_j + B_{k+j}\,\vvec_j,\qquad j=0,\ldots,N{-}1.
\end{equation}
The initial condition is set to the current measured state, $\zvec_0=\xvec_k$, so that the nominal and true trajectories coincide at the start of each planning cycle.

\subsection{Error Dynamics and the Tube Feedback Law}

The control actually applied to the spacecraft is not the nominal input alone, but a sum of the planned input and a feedback correction:
\begin{equation}\label{eq:tube_law}
  \uvec_k = \vvec_k + K(\xvec_k - \zvec_k).
\end{equation}
Subtracting the nominal dynamics~\eqref{eq:nom_dyn} from the true dynamics~\eqref{eq:dynamics} and substituting the control law~\eqref{eq:tube_law} yields the \emph{error dynamics}:
\begin{equation}\label{eq:error_dyn}
  \evec_{k+1} = (A_k + B_k K)\,\evec_k + \wvec_k = A_{\mathrm{cl},k}\,\evec_k + \wvec_k,
\end{equation}
where $\evec_k=\xvec_k-\zvec_k$ is the tracking error and $A_{\mathrm{cl},k}=A_k+B_kK$ is the closed-loop error matrix at step $k$. The error evolves as a stable linear system driven by the disturbance; the feedback gain $K$ determines how quickly the error decays, and hence how wide the tube must be.

The gain $K$ is designed once, offline, via discrete LQR at a representative operating point:
\begin{equation}\label{eq:gain}
  K = -\texttt{dlqr}(A_0,\,B_0,\,Q_K,\,R_K),
\end{equation}
with $Q_K=\mathrm{diag}(10^3,10^3,10^3,10,10,10)$ and $R_K=I_3$. The large position weights in $Q_K$ (relative to $R_K$) produce a gain that aggressively suppresses position errors, keeping the tube narrow in the spatial dimensions that matter most for constraint satisfaction. Using a separate cost pair $(Q_K,R_K)$ for the tube gain, distinct from the MPC cost $(Q,R)$, decouples two design objectives: $K$ controls the tube width (robustness), while $Q$ and $R$ shape the nominal trajectory (optimality).

\subsection{Element-Wise Box Tightening for LTV Systems}
\label{sec:emax}

The recursion rests on one elementary fact about element-wise absolute values.

\begin{lemma}[Element-Wise Triangle Inequality]\label{lem:abs}
Let $M\in\R^{m\times n}$, $v\in\R^n$, and let $\bar v\in\R^n_{\ge 0}$ satisfy $|v|\le \bar v$ element-wise. Then
\begin{equation}\label{eq:absM}
  |Mv|\;\le\;|M|\,|v|\;\le\;|M|\,\bar v\qquad\text{element-wise},
\end{equation}
where $|\cdot|$ denotes element-wise absolute value as defined in \emph{\nameref{sec:notation}}.
\end{lemma}

\begin{proof}
For each row $i$, the standard triangle inequality on scalar sums gives $(|Mv|)_i = |\sum_l M_{il}v_l| \le \sum_l |M_{il}|\,|v_l| = (|M|\,|v|)_i$, establishing the first inequality. Because every entry of $|M|$ is non-negative, the map $v\mapsto |M|v$ is monotone with respect to the element-wise order: $|v|\le\bar v$ implies $|M|\,|v|\le|M|\,\bar v$, giving the second inequality.
\end{proof}

\begin{remark}[How \Cref{lem:abs} controls the input-side tightening]\label{rem:Klem}
Applied with $M=K\in\R^{3\times 6}$ and $v=\evec_j$, \Cref{lem:abs} yields $|K\evec_j|\le |K|\,|\evec_j|\le |K|\,\emax_j$ element-wise. This bound is the sole justification for the input tightening $u_{\min}+|K|\emax_j\le\vvec_j\le u_{\max}-|K|\emax_j$ that appears later in~\eqref{eq:tubeMPC_ucon}. The bound is valid at every horizon step $j\in\{0,\dots,N-1\}$ for which the error-bound recursion has been propagated (\Cref{prop:ebound} below). Consequently the input tightening uses the \emph{same} $\emax_j$ sequence as the state tightening, ensuring a single internally consistent bounding logic across both constraint classes.
\end{remark}

Define the element-wise error bound $\emax_j\in\R^6_{\geq0}$ by the recursion
\begin{equation}\label{eq:emax_init}
  \emax_0 = \mathbf{0},
\end{equation}
\begin{equation}\label{eq:emax_prop}
  \emax_{j+1} = |A_{\mathrm{cl},k+j}|\,\emax_j + \wmax_j,\qquad j=0,\ldots,N{-}1,
\end{equation}
where $|A_{\mathrm{cl},k+j}|$ is the matrix whose entries are the absolute values of those of $A_{\mathrm{cl},k+j}$, and $\wmax_j$ is a \emph{per-step} disturbance bound defined below. This recursion costs $\mathcal{O}(Nn_x^2)$ floating-point operations, negligible compared with the QP solve.

\subsubsection{Mass-Aware Per-Step Disturbance Bound}
\label{sec:mass_aware}

Thrust mismatch due to mass uncertainty introduces a multiplicative perturbation $(1+\delta_m)B_k$ with $|\delta_m|\le\delta_{m,\max}$. The induced disturbance is bounded element-wise by $|\delta_m B_k\uvec_k|\le\delta_{m,\max}|B_k|u_{\max}$ (per step, because $B_k$ varies with $\nu$). Combining with $\wmax$ gives a per-step total
\begin{equation}\label{eq:wmax_total}
  \wmax_j = \wmax + \delta_{m,\max}\,|B_{k+j}|\,u_{\max},\qquad j=0,\ldots,N{-}1,
\end{equation}
absorbing multiplicative uncertainty into the additive-tube framework without a separate tube. At the medium tier ($\delta_{m,\max}=5\%$) the mass contribution adds $50$--$60$~m per step, making total tightening $\sim 14\%$ of the $\pm 500$~m corridor.

Three properties make the recursion useful: a spectral-radius condition for the contraction (\Cref{lem:stability}), validity of the bound (\Cref{prop:ebound}), and convergence to a finite steady state (\Cref{prop:convergence}). \Cref{lem:stability} certifies that the non-negative matrix $|A_{\mathrm{cl},k}|$ is Schur-stable in the Perron--Frobenius sense, the condition the recursion~\eqref{eq:emax_prop} requires for $\emax_j$ to remain bounded.

\begin{lemma}[Closed-Loop Stability via Perron--Frobenius]\label{lem:stability}
Under \Cref{ass:dynamics,ass:controllable}, let $K$ be obtained from~\eqref{eq:gain}. Then:
\begin{enumerate}
\item[\emph{(i)}] $A_{\mathrm{cl},0}=A_0+B_0K$ is Schur stable: $\rho(A_{\mathrm{cl},0})<1$.
\item[\emph{(ii)}] The non-negative matrix $|A_{\mathrm{cl},0}|$ satisfies $\rho(|A_{\mathrm{cl},0}|)<1$.
\item[\emph{(iii)}] There exists a computable eccentricity threshold $e_{\max}>0$ such that for all orbits with $e<e_{\max}$ and for all $k$:
\begin{equation}\label{eq:spectral_bound}
  \rho(|A_{\mathrm{cl},k}|) < 1.
\end{equation}
\end{enumerate}
\end{lemma}

\begin{proof}
(i) follows from standard LQR theory~\cite{Rawlings_Mayne}: $(A_0,B_0)$ controllable plus $Q_K,R_K\succ 0$ yield $\rho(A_{\mathrm{cl},0})<1$. (ii) does \emph{not} follow from (i). Only $\rho(|M|)\ge\rho(M)$ holds in general~\cite{Horn_Johnson}, which bounds $\rho(M)$ from above and says nothing about $\rho(|M|)$; the reverse implication is false for matrices whose stability relies on sign cancellation (\Cref{rem:pathology}). Condition (ii) is therefore a \emph{verifiable design condition}, not a consequence: it is checked numerically at the design point in \Cref{rem:spec_numerical}, where it holds with two orders of magnitude of margin. (iii) is a continuity argument over $e$: the entries of $A_{\mathrm{cl},k}$ depend continuously on $e$ and $\rho$ is continuous in the entries, so the set $\{e:\max_k\rho(|A_{\mathrm{cl},k}(e)|)<1\}$ is open and contains the verified design point; $e_{\max}$ is defined as the smallest $e$ above the design point at which $\max_k\rho(|A_{\mathrm{cl},k}|)=1$ and is located by bisection in \Cref{rem:spec_numerical}. Monotonicity of $\max_k\rho$ in $e$ is observed numerically over the sampled range (\Cref{fig:emax_threshold}) but is not proved here, so $e_{\max}$ should be read as the numerically located first crossing.
\end{proof}

\begin{remark}[Numerical verification of \Cref{lem:stability}]\label{rem:spec_numerical}
For the MSRE design point ($e=0.204$, $\nu_0=0$, $Q_K, R_K$ as in \Cref{tab:params}), direct evaluation gives $\rho(A_{\mathrm{cl},0})=0.0016$, $\rho(|A_{\mathrm{cl},0}|)=0.0019$, and $\|A_{\mathrm{cl},0}\|_\infty=0.006$, so $\delta_0\approx 0.998$. Sampling 200 equispaced true-anomaly values yields $\max_k\rho(|A_{\mathrm{cl},k}|)=0.128$ (\Cref{fig:spectral_radius_orbit}) and, for the element-wise maximum required by \Cref{prop:convergence}, $\rho(\bar A)=0.129$. The eccentricity threshold of part~(iii) is obtained by bisection: at each candidate $e$, $K$ is recomputed via LQR at $\nu=0$ and $\max_\nu\rho(|A_{\mathrm{cl}}(\nu)|)$ is evaluated over the same 200-point grid. The crossing $\max_\nu\rho(|A_{\mathrm{cl}}(\nu)|)=1$ occurs at $e_{\max}\approx 0.661$, giving a safety factor $e_{\max}/e_{\mathrm{MSRE}}\approx 3.23$ (\Cref{fig:emax_threshold}). These numerics are application-specific and do not enter the proof; they confirm that \Cref{lem:stability} is non-vacuous for the MSRE scenario.
\end{remark}

\begin{figure}[htbp]
  \centering
  \includegraphics[width=0.7\textwidth]{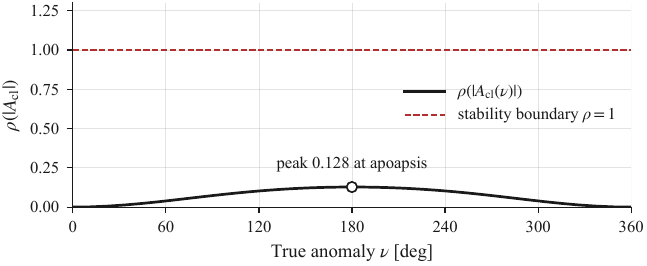}
  \caption{Spectral radius $\rho(|A_{\mathrm{cl},k}|)$ versus true anomaly $\nu$ on the MSRE orbit. The value remains below unity at all orbit phases, with a peak of 0.128 near periapsis where the orbital dynamics vary most rapidly.}
  \label{fig:spectral_radius_orbit}
\end{figure}

\begin{figure}[htbp]
  \centering
  \includegraphics[width=0.7\textwidth]{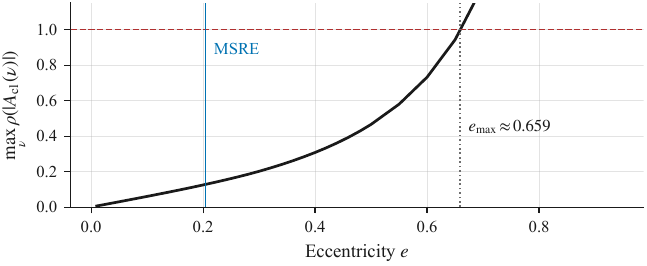}
  \caption{Maximum spectral radius $\max_\nu\rho(|A_{\mathrm{cl}}(\nu)|)$ as a function of eccentricity. The dashed red line marks the stability boundary at~1. The MSRE orbit ($e=0.204$) lies well within the feasible region.}
  \label{fig:emax_threshold}
\end{figure}

\begin{remark}\label{rem:pathology}
The condition $\rho(|A_{\mathrm{cl},k}|)<1$ is strictly stronger than Schur stability. A Schur-stable matrix can have $\rho(|M|)>1$ if it contains large entries with cancelling signs; for instance, a $2\times2$ rotation near $\pi/2$ has spectral radius less than one but element-wise absolute value with spectral radius exceeding one. The LQR design with dominant position weighting avoids this pathology by producing gains that make $A_{\mathrm{cl}}$ element-wise small rather than relying on sign cancellations.
\end{remark}

\Cref{prop:ebound} states soundness: $\emax_j$ upper-bounds the realised tracking error $|\evec_j|$ at every horizon step for every admissible disturbance sequence.

\begin{proposition}[Error Bound Validity]\label{prop:ebound}
Under \Cref{ass:bounded_w}, if $|\evec_0|\leq\emax_0$ element-wise and $|\wvec_j|\leq\wmax$ for all $j$, then
\begin{equation}\label{eq:ebound}
  |\evec_j|\leq\emax_j,\qquad j=0,1,\ldots,N,
\end{equation}
where $\emax_j$ satisfies the recursion~\eqref{eq:emax_prop}.
\end{proposition}

\begin{proof}
By induction on $j$. The base case $|\evec_0|=|\xvec_k-\zvec_0|=\mathbf{0}=\emax_0$ is immediate since $\zvec_0=\xvec_k$. For the inductive step, \Cref{lem:abs} applied to $M=A_{\mathrm{cl},k+j}$ and $v=\evec_j$ together with the inductive hypothesis $|\evec_j|\le\emax_j$ gives
\begin{equation}\label{eq:abs_ineq}
  |\evec_{j+1}| = |A_{\mathrm{cl},k+j}\evec_j+\wvec_j| \leq |A_{\mathrm{cl},k+j}|\,|\evec_j|+|\wvec_j| \leq |A_{\mathrm{cl},k+j}|\,\emax_j+\wmax_j = \emax_{j+1},
\end{equation}
using \Cref{ass:bounded_w} for $|\wvec_j|\le\wmax\le\wmax_j$ (since $\wmax_j=\wmax+\delta_{m,\max}|B_{k+j}|u_{\max}\ge\wmax$).
\end{proof}

\Cref{prop:convergence} removes the horizon-length dependence: when the orbit-worst-case matrix is Perron--Frobenius stable the recursion has a finite steady state $\emax_\infty$, the ``cost of robustness'' that quantifies the corridor width consumed by the tube. Its hypothesis $\rho(\bar A)<1$ is a condition on the element-wise maximum $\bar A$, which is strictly stronger than $\rho(|A_{\mathrm{cl},k}|)<1$ holding at each $k$ separately, and must be verified in its own right; for the MSRE design $\rho(\bar A)=0.129$ against $\max_k\rho(|A_{\mathrm{cl},k}|)=0.128$ (\Cref{rem:spec_numerical}).

\begin{proposition}[Convergence of Error Bound]\label{prop:convergence}
If the orbit-worst-case matrix $\bar A=\max_k|A_{\mathrm{cl},k}|$ satisfies $\rho(\bar A)<1$, then the error bound sequence $\{\emax_j\}_{j\geq0}$ converges to a finite limit satisfying
\begin{equation}\label{eq:ess}
  \emax_\infty \leq (I-\bar{A})^{-1}\wmax_\infty,\qquad \wmax_\infty:=\wmax+\delta_{m,\max}\max_k|B_k|\,u_{\max},
\end{equation}
where $\bar{A}=\max_k|A_{\mathrm{cl},k}|$ is the element-wise maximum over one orbital period.
\end{proposition}

\begin{proof}
$|A_{\mathrm{cl},k+j}|\le\bar A$ and $\wmax_j\le\wmax_\infty$ hold element-wise by construction. Consider the majorising recursion $\hat{\evec}_{j+1}=\bar{A}\hat{\evec}_j+\wmax_\infty$ from $\hat{\evec}_0=\mathbf{0}$. Its closed form is $\hat{\evec}_j=\sum_{i=0}^{j-1}\bar{A}^i\,\wmax_\infty$; since $\rho(\bar A)<1$, the Neumann series converges to $(I-\bar A)^{-1}\wmax_\infty<\infty$. The bound $\emax_j\le\hat{\evec}_j$ follows by induction using the same step as in \Cref{prop:ebound}, with $\bar A$ replacing $|A_{\mathrm{cl},k+j}|$ and $\wmax_\infty$ replacing $\wmax_j$. Taking $j\to\infty$ yields~\eqref{eq:ess}.
\end{proof}

\begin{remark}
The horizon-dependent bound is less conservative than the constant-width mRPI tube because (i)~$\emax_0=\mathbf{0}$ leaves the first prediction step untightened, and (ii)~the actual matrices $A_{\mathrm{cl},k+j}$ track orbit-phase-dependent error growth rather than bounding it uniformly with $\bar{A}$. \Cref{fig:tube_tightening} illustrates the tightening profile.
\end{remark}

\begin{figure}[htbp]
  \centering
  \includegraphics[width=\textwidth]{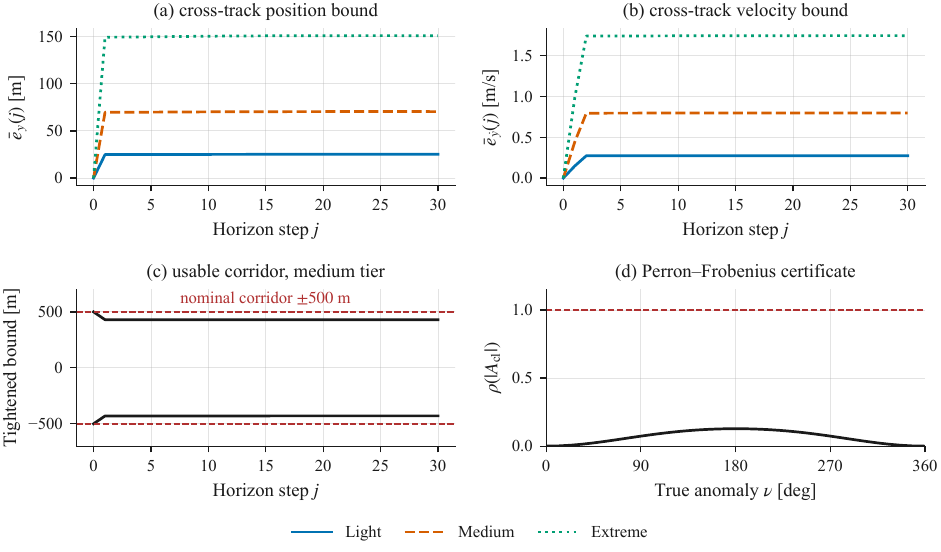}
  \caption{\textbf{Horizon-dependent constraint tightening on the MSRE orbit.}
  Light (solid), medium (dashed), extreme (dotted) tiers.
  (a)~Cross-track position bound $\bar e_{\mathrm{max},y}(j)$ grows from zero
  and saturates in $\sim 5$ steps; saturation level scales with the tier.
  (b)~Cross-track velocity bound, same staircase pattern.
  (c)~Tightened corridor (shaded) vs original $\pm 500$~m (thin lines): early
  steps retain almost the full corridor; late steps lose $\sim 14$\%
  (medium) to $\sim 30$\% (extreme) of half-width.
  (d)~Spectral radius $\rho(|A_{\mathrm{cl},k}|)$ over the orbit, below unity
  (dashed red) throughout, the Perron--Frobenius certificate of
  \Cref{lem:stability}.}
  \label{fig:tube_tightening}
\end{figure}

\subsection{Tightened-Constraint QP}

The tube MPC QP replaces the original constraints in~\eqref{eq:nomMPC} with their tightened counterparts:
\begin{subequations}\label{eq:tubeMPC}
\begin{align}
  \min_{\zvec_{0:N},\,\vvec_{0:N-1}} \quad & \sum_{j=0}^{N-1}\!\Big[(\zvec_j-\xvec_s)^\top Q\,(\zvec_j-\xvec_s) + \vvec_j^\top R\,\vvec_j\Big] + (\zvec_N-\xvec_s)^\top P\,(\zvec_N-\xvec_s) \label{eq:tubeMPC_cost}\\
  \text{s.t.}\quad & \zvec_0 = \xvec_k, \label{eq:tubeMPC_init}\\
  & \zvec_{j+1}=A_{k+j}\,\zvec_j + B_{k+j}\,\vvec_j,\quad j=0,\ldots,N{-}1, \label{eq:tubeMPC_dyn}\\
  & x_{\min}+\emax_j\leq\zvec_j\leq x_{\max}-\emax_j,\quad j=1,\ldots,N, \label{eq:tubeMPC_xcon}\\
  & u_{\min}+|K|\,\emax_j\leq\vvec_j\leq u_{\max}-|K|\,\emax_j,\quad j=0,\ldots,N{-}1. \label{eq:tubeMPC_ucon}
\end{align}
\end{subequations}

The state tightening~\eqref{eq:tubeMPC_xcon} shrinks the position and velocity corridors by $\emax_j$ on each side, ensuring that even in the worst case the true state $\xvec_j=\zvec_j+\evec_j$ remains inside $[x_{\min},x_{\max}]$. The input tightening~\eqref{eq:tubeMPC_ucon} accounts for the feedback correction $K\evec_j$, whose worst-case magnitude is $|K|\emax_j$ entry-wise.

For the tightened sets to be non-empty, one requires
\begin{equation}\label{eq:feasibility_check}
  x_{\min}+\emax_j < x_{\max}-\emax_j\quad\text{and}\quad u_{\min}+|K|\emax_j < u_{\max}-|K|\emax_j
\end{equation}
for all $j=0,\ldots,N$. If any tightened set is empty, the disturbance bound $\wmax$ exceeds what the system can robustly tolerate given the constraint margins, and the controller falls back to LQR.

Because $\emax_0=\mathbf{0}$, the tightened constraints at $j=0$ coincide with the original constraints. The tightening grows monotonically with $j$, so the constraint set is narrowest at the end of the horizon. The applied control is simply $\uvec_k=\vvec_0^\star$: since $\zvec_0=\xvec_k$, the tracking error is zero at $j=0$, and the feedback correction $K\evec_0=\mathbf{0}$ vanishes.

\subsection{Robustness Guarantees}

Three guarantees build on \Cref{prop:ebound}: robust constraint satisfaction, recursive feasibility, and an input-to-state-stability bound. Proofs longer than four lines are deferred to \Cref{app:terminal_set}.

\begin{theorem}[Robust Constraint Satisfaction]\label{thm:robust}
Under \Cref{ass:dynamics,ass:bounded_w,ass:controllable}, if the tube MPC QP~\eqref{eq:tubeMPC} is feasible at time $k$, then for any disturbance realisation satisfying $|\wvec_j|\leq\wmax$, the true closed-loop state $\xvec_{k+j}$ and input $\uvec_{k+j}$ satisfy the original constraints~\eqref{eq:xconstraint}--\eqref{eq:uconstraint} for $j=0,\ldots,N$.
\end{theorem}

\begin{proof}
By \Cref{prop:ebound}, $|\evec_j|\le\emax_j$. The state decomposition $\xvec_{k+j}=\zvec_j+\evec_j$ with the tightened state constraint~\eqref{eq:tubeMPC_xcon} gives $\zvec_j+\emax_j\le x_{\max}$ and $\zvec_j-\emax_j\ge x_{\min}$, so $x_{\min}\le\xvec_{k+j}\le x_{\max}$. The same argument applied to $\uvec_{k+j}=\vvec_j+K\evec_j$ with $|K\evec_j|\le|K|\emax_j$ and the tightened input constraint~\eqref{eq:tubeMPC_ucon} delivers $u_{\min}\le\uvec_{k+j}\le u_{\max}$.
\end{proof}

Recursive feasibility (feasibility at $k$ implies feasibility at $k{+}1$ despite unknown disturbances) requires one terminal ingredient: an ellipsoidal robust positively invariant (RPI) terminal set under the static feedback gain $K$, enforced as a single quadratic constraint in the QP. The construction is summarised below; the derivation of the level set $\alpha$ is given in \Cref{app:terminal_set}.

\paragraph{Ellipsoidal terminal set.}
Let $\nu_\star$ denote the worst-case orbit phase that maximises the contraction-loss ratio $\lambda_{\max}(P^{-1/2}A_{\mathrm{cl},k}^\top P A_{\mathrm{cl},k} P^{-1/2})$ over one orbital period (numerically, $\nu_\star$ corresponds to apocenter for the MSRE design), and let $P=\mathrm{DARE}(A_\star,B_\star,Q_K{+}\varepsilon I,R_K)$ be the Riccati solution there. Define
\begin{equation}\label{eq:Xf}
  \mathcal{X}_f \;=\; \{\xvec\in\R^6:\;(\xvec-\xvec_s)^\top P\,(\xvec-\xvec_s)\;\le\;\alpha\},
\end{equation}
with the level $\alpha$ chosen as
\begin{equation}\label{eq:alpha}
  \alpha \;=\; \frac{(1+1/\varepsilon_*)\,\wmax_\infty^\top |P|\,\wmax_\infty}{1-(1+\varepsilon_*)\,\lambda_\star},
  \qquad \lambda_\star := \max_k \lambda_{\max}\!\left(P^{-1/2}A_{\mathrm{cl},k}^\top P A_{\mathrm{cl},k} P^{-1/2}\right),
\end{equation}
where $\varepsilon_*\in(0,(1-\lambda_\star)/\lambda_\star)$ is a free parameter chosen offline to minimise $\alpha$ (\Cref{app:terminal_set}). The inequality $\lambda_\star<1$ is equivalent to $P$ being a common Lyapunov function for the family $\{A_{\mathrm{cl},k}\}_k$ on the orbit; \Cref{lem:terminal} verifies this is satisfied under \Cref{lem:stability}.

\begin{lemma}[Robust Forward Invariance of $\mathcal{X}_f$]\label{lem:terminal}
Under \Cref{ass:dynamics,ass:bounded_w,ass:controllable} and the conditions of \Cref{lem:stability}, the set $\mathcal{X}_f$ defined by~\eqref{eq:Xf}--\eqref{eq:alpha} is robustly positively invariant under the closed-loop dynamics $\xvec_{k+1}=A_{\mathrm{cl},k}(\xvec_k-\xvec_s)+\xvec_s+\wvec_k$: for every $k$, every $\xvec_k\in\mathcal{X}_f$, and every $\wvec_k$ with $|\wvec_k|\le\wmax_\infty$, $\xvec_{k+1}\in\mathcal{X}_f$.
\end{lemma}

The proof, the offline choice of $\varepsilon_*$, and the level-set arithmetic are deferred to \Cref{app:terminal_set}.

\paragraph{Terminal constraint and augmented QP.}
The tube MPC QP~\eqref{eq:tubeMPC} is augmented with the single quadratic constraint
\begin{equation}\label{eq:Xf_QP}
  (\zvec_N-\xvec_s)^\top P\,(\zvec_N-\xvec_s)\;\le\;\alpha_N,\qquad \alpha_N:=\Big(\sqrt{\alpha}-\big\|\,|P^{1/2}|\,\emax_N\big\|_2\Big)^2,
\end{equation}
where the deflated level $\alpha_N$ ensures that the predicted nominal terminal state $\zvec_N$ lies in $\mathcal{X}_f$ \emph{after} accounting for the worst-case error tube radius $\emax_N$ in the $P$-metric. The element-wise absolute value $|P^{1/2}|$ is required: $|\evec|\le\emax_N$ bounds the error component-wise, and $\|P^{1/2}\evec\|_2\le\|\,|P^{1/2}|\,\emax_N\|_2$ holds for every such $\evec$, whereas $\|P^{1/2}\emax_N\|_2$ does not upper-bound the box when $P^{1/2}$ carries negative entries. For the MSRE design the two differ by less than $10^{-5}$ relative, since $P$ is dominated by $Q_K$. \Cref{eq:Xf_QP} is a second-order cone constraint. OSQP solves quadratic programs with linear constraints only, so an implementation on that solver must replace it with a polyhedral surrogate; the accompanying code uses the axis-aligned box $|(\zvec_N-\xvec_s)_i|\le\sqrt{\alpha_N/P_{ii}}$, which \emph{circumscribes} the ellipsoid and is therefore a relaxation, not an enforcement, of~\eqref{eq:Xf_QP}. An inscribed polytope, or a solver with native SOC support, is required for \Cref{thm:recursive} to transfer to the implementation. The Monte Carlo campaign of \Cref{sec:results} was run without the terminal constraint; see \Cref{rem:terminal_scope}.

With~\eqref{eq:Xf_QP} in place, recursive feasibility becomes a theorem rather than a sufficient-condition template.

\begin{theorem}[Recursive Feasibility]\label{thm:recursive}
Under \Cref{ass:dynamics,ass:bounded_w,ass:controllable} and \Cref{lem:stability,lem:terminal}, suppose the augmented tube MPC QP~\eqref{eq:tubeMPC}+\eqref{eq:Xf_QP} is feasible at time $k$. For every disturbance realisation with $|\wvec_k|\le\wmax$, the augmented QP at time $k{+}1$ is feasible.
\end{theorem}

The candidate is a shift-and-append: $\tilde{\vvec}_j=\vvec_{j+1}^\star$ for $j<N{-}1$, $\tilde{\vvec}_{N-1}=K(\tilde\zvec_{N-1}-\xvec_s)$. Reinitialising $\tilde\emax_0=\mathbf 0$ at $\tilde\zvec_0=\xvec_{k+1}$ inherits linear-constraint feasibility because $\tilde\emax_j\le\emax_{j+1}$; the terminal constraint follows from $\zvec_N^\star\in\mathcal{X}_f$ together with \Cref{lem:terminal}. Full algebra is in \Cref{app:terminal_set}.

\begin{remark}[Scope of \Cref{thm:recursive} in this study]\label{rem:terminal_scope}
\Cref{thm:recursive} is a design result: it holds for the QP that enforces~\eqref{eq:Xf_QP} exactly. The Monte Carlo campaign reported in \Cref{sec:results} was run with the unaugmented tube MPC of~\eqref{eq:tubeMPC}, without the terminal constraint, so none of the empirical results below exercises the terminal set, and the recursive-feasibility guarantee is not what produces the reported fallback rates. Those rates are governed by the linear tightening and by triggers F1 and F3 of \Cref{sec:infeasibility}. Quantifying what the terminal constraint adds in closed loop, using a solver with native second-order-cone support, is left to future work.
\end{remark}

\begin{remark}[Why the worst-case design point?]\label{rem:Pdesign}
Computing $P$ at the apocenter true anomaly $\nu_\star$ (rather than at $\nu_0=0$) ensures that $\lambda_\star<1$ in~\eqref{eq:alpha} is satisfied at every orbit phase, since by construction $\nu_\star$ maximises the contraction-loss ratio. The MSRE evaluation gives $\lambda_\star=0.016$, $\varepsilon_*=6.90$, and $\alpha\approx 2.3\times 10^{7}$ (units consistent with the $Q_K$ weighting) at the medium tier. The resulting $\mathcal{X}_f$ projected onto the $(r_y,r_z)$ plane covers an ellipse of semi-axes $\approx[152,\,152]$~m, comfortably inside the $\pm 500$~m corridor.
\end{remark}

\subsection{Conditions for QP Infeasibility and Recovery}
\label{sec:infeasibility}

Three operational triggers can still cause infeasibility despite \Cref{thm:recursive}:
\begin{enumerate}[label=(F\arabic*),leftmargin=2.5em,nosep]
  \item \textbf{Empty linear tightened sets:} $x_{\min}+\emax_j>x_{\max}-\emax_j$ or $u_{\min}+|K|\emax_j>u_{\max}-|K|\emax_j$ at some $j$; detected before OSQP. Indicates that $\wmax$, $K$, or the corridor is mismatched.
  \item \textbf{Empty terminal level set:} $\alpha_N\le 0$ in~\eqref{eq:Xf_QP}, equivalent to (F1) localised to the terminal constraint.
  \item \textbf{Out-of-bound disturbance:} $\Delta\xvec_0$ or $\wvec_k$ exceeds $\wmax$ (e.g.\ off-nominal $J_2$ residuals), landing the realised state outside the feasible basin.
\end{enumerate}
On any trigger the controller applies $\uvec_k=\mathrm{sat}(-K_{\mathrm{LQR}}(\xvec_k-\xvec_s),u_{\max})$, re-propagates $\emax_j$ from zero, and re-attempts the QP at $k{+}1$. The \Cref{thm:robust} certificate is forfeited for the fallback step; closed-loop stability via \Cref{lem:stability} and the bound of \Cref{prop:ebound} persist. \Cref{tab:fallback} reports $0.7$ fallbacks per trial at medium tier and $9.5$ at extreme.

\Cref{cor:iss} converts the bound of \Cref{prop:convergence} into a standard ISS gain in the $\infty$-norm.

\begin{corollary}[Input-to-State Stability]\label{cor:iss}
Under \Cref{thm:robust,thm:recursive} and the conditions of \Cref{lem:stability}, the closed-loop system~\eqref{eq:dynamics} under the tube MPC law~\eqref{eq:tube_law} is input-to-state stable with respect to the disturbance $\wvec$. The tracking error satisfies
\begin{equation}\label{eq:iss}
  |\evec_j|\leq\bar{A}^j\,|\evec_0| + \sum_{i=0}^{j-1}\bar{A}^i\,\wmax_\infty \;\leq\; (I-\bar{A})^{-1}\wmax_\infty,
\end{equation}
with $\bar{A}=\max_k|A_{\mathrm{cl},k}|$ and $\rho(\bar{A})<1$ as verified in \Cref{rem:spec_numerical}. The asymptotic ISS gain is
\begin{equation}\label{eq:iss_gain}
  \gamma_{\mathrm{ISS}} = \|(I-\bar{A})^{-1}\|_\infty.
\end{equation}
\end{corollary}

\begin{proof}
Iterating the error dynamics~\eqref{eq:error_dyn} and applying \Cref{lem:abs} step by step gives the first inequality in~\eqref{eq:iss}; the second follows from $\rho(\bar A)<1$ and the Neumann series argument of \Cref{prop:convergence}. Taking the $\infty$-norm yields $\|\evec_j\|_\infty\le\gamma_{\mathrm{ISS}}\|\wmax_\infty\|_\infty$.
\end{proof}

\begin{remark}[Numerical values of the ISS bound for the MSRE design]\label{rem:iss_numerical}
For the MSRE design point, $\gamma_{\mathrm{ISS}}\approx 1.16$. With the mass-aware per-step bound~\eqref{eq:wmax_total} at the medium tier, the steady-state error bound evaluates to $\emax_\infty\approx[90.7,\,70.5,\,92.7,\,1.10,\,0.80,\,1.12]^\top$~(m;m/s), tightening the $\pm 500$~m lateral corridor by $14\%$ on the cross-track side and $19\%$ on the radial side. The horizon-dependent recursion reaches a smaller terminal value, $\emax_N\approx[84,\,71,\,86]^\top$~m at $N=30$, because it propagates the actual $|A_{\mathrm{cl},k+j}|$ rather than the orbit-worst-case $\bar A$; the gap between the two is quantified in \Cref{sec:conservatism_compare_mc}.
\end{remark}

\subsection{Summary of Theoretical Guarantees}
\label{sec:cert_guarantees}

\Cref{tab:guarantees} consolidates the formal results of \Cref{sec:tube}.

\begin{table}[htbp]
\centering
\caption{Guarantees provided by the augmented tube MPC of \Cref{sec:tube}. ``Holds'' indicates that the property is established under the listed assumptions.}
\label{tab:guarantees}
\small
\begin{tabular}{@{}l l l >{\raggedright\arraybackslash}p{2.5cm}@{}}
\toprule
Property & Result & Required assumptions & Holds \\
\midrule
Element-wise input bound $|K\evec_j|\le|K|\emax_j$ & \Cref{lem:abs} & none & always \\
Spectral-radius condition $\rho(|A_{\mathrm{cl},k}|)<1$ & \Cref{lem:stability} & \cref{ass:dynamics,ass:controllable}; $e<e_{\max}$ & yes \\
Error-bound validity $|\evec_j|\le\emax_j$ & \Cref{prop:ebound} & \cref{ass:bounded_w}; \Cref{lem:abs} & yes \\
Steady-state convergence $\emax_\infty<\infty$ & \Cref{prop:convergence} & $\rho(\bar A)<1$ & yes \\
Robust constraint satisfaction & \Cref{thm:robust} & \Cref{prop:ebound}; QP feasible & yes \\
Robust forward invariance of $\mathcal{X}_f$ & \Cref{lem:terminal} & \Cref{lem:stability}; $\lambda_\star<1$ & yes \\
Recursive feasibility & \Cref{thm:recursive} & \Cref{lem:terminal}; \eqref{eq:Xf_QP} enforced exactly & by design; not exercised here \\
Input-to-state stability $\gamma_{\mathrm{ISS}}<\infty$ & \Cref{cor:iss} & \Cref{thm:robust,thm:recursive} & yes \\
\bottomrule
\end{tabular}
\end{table}

\paragraph{Limits of the guarantees.} The \Cref{thm:robust} certificate holds only at feasible-QP steps; at fallback steps it is forfeited. The additive disturbance budget $\wmax_\infty$ must cover realised perturbations; structural model errors beyond~\eqref{eq:wmax_total} are out of scope. Timing is measured, not certified.

\subsection{Implementation: Algorithm and Solver}
\label{sec:solver}

\begin{algorithm}[htbp]
\caption{Augmented Tube-Based MPC for YA Relative Motion (with Terminal Constraint)}
\label{alg:tube}
\begin{algorithmic}[1]
\REQUIRE State $\xvec_k$, true anomaly $\nu_k$, parameters $Q,R,P,K,\wmax,\delta_{m,\max},\alpha$
\STATE Build $\{A_{k+j},B_{k+j}\}_{j=0}^{N-1}$ via YA STM at $\nu_k,T_s$
\STATE $A_{\mathrm{cl},j}\gets A_{k+j}+B_{k+j}K$ for $j=0,\ldots,N{-}1$
\STATE $\emax_0\gets\mathbf{0}$
\FOR{$j=0$ \TO $N{-}1$}
  \STATE $\wmax_j\gets\wmax+\delta_{m,\max}\,|B_{k+j}|\,u_{\max}$ \COMMENT{mass-aware bound~\eqref{eq:wmax_total}}
  \STATE $\emax_{j+1}\gets|A_{\mathrm{cl},j}|\,\emax_j+\wmax_j$
\ENDFOR
\STATE Tighten: $x_{\min/\max,j}^{\mathrm{t}}\gets x_{\min/\max}\pm\emax_j$, $u_{\min/\max,j}^{\mathrm{t}}\gets u_{\min/\max}\pm|K|\emax_j$
\STATE \textbf{if} terminal set enabled \textbf{then} $\alpha_N\gets\big(\sqrt\alpha-\big\||P^{1/2}|\,\emax_N\big\|_2\big)^2$ \COMMENT{optional; disabled in \Cref{sec:results}, see \Cref{rem:terminal_scope}}
\IF{any tightened linear set is empty \textbf{or} $\alpha_N\le 0$}\label{alg:f1f2}
  \STATE Fallback: $\uvec_k\gets\mathrm{sat}(-K_{\mathrm{LQR}}(\xvec_k-\xvec_s),\,u_{\max})$ \COMMENT{trigger F1/F2; see \Cref{sec:infeasibility}}
\ELSE
  \STATE Solve augmented QP~\eqref{eq:tubeMPC}+\eqref{eq:Xf_QP} via OSQP $\to\vvec_0^\star$
  \IF{QP infeasible}\label{alg:f3}
    \STATE Fallback: $\uvec_k\gets\mathrm{sat}(-K_{\mathrm{LQR}}(\xvec_k-\xvec_s),\,u_{\max})$ \COMMENT{trigger F3}
  \ELSE
    \STATE $\uvec_k\gets\vvec_0^\star$ \COMMENT{certified by \Cref{thm:robust}}
  \ENDIF
\ENDIF
\RETURN $\uvec_k$
\end{algorithmic}
\end{algorithm}

The QP is solved in condensed form: eliminating states via~\eqref{eq:tubeMPC_dyn} leaves $3N=90$ decision variables and $276$ linear inequalities, all handled natively by OSQP~\cite{OSQP}. Settings: max iter $20{,}000$, tol $10^{-7}$, polish off, warm-start on. Measured solve time on a 3.5~GHz desktop is $29.3$~ms (mean), $4.0$~ms (median), $549$~ms (worst) over $18{,}000$ solves, all well within $T_s=200$~s. Infeasibility (\Cref{sec:infeasibility}, F1--F3) falls back to saturated LQR; OSQP's auto-generated C provides a path to embedded deployment.

\section{Simulation Setup}
\label{sec:setup}

\subsection{MSRE Orbit}

The simulations model a rendezvous campaign on the Mars Sample Return Elliptical (MSRE) orbit, whose parameters reflect a representative Mars orbit insertion scenario. \Cref{tab:orbit} summarises the orbital elements.

\begin{table}[htbp]
\centering
\caption{MSRE orbital and scenario parameters.}
\label{tab:orbit}
\begin{tabular}{@{}llc@{}}
\toprule
Parameter & Symbol & Value \\
\midrule
Central body & & Mars \\
Gravitational parameter & $\mu$ & $4.2835\times10^{13}$~m$^3$/s$^2$ \\
Mars equatorial radius & $R_M$ & \SI{3396.2}{\kilo\metre} \\
Semi-major axis & $a$ & \SI{4643}{\kilo\metre} \\
Eccentricity & $e$ & 0.2044 \\
Orbital period & $T_{\mathrm{orb}}$ & \SI{9605}{\second} ($\approx$160~min) \\
Mean motion & $n$ & $6.542\times10^{-4}$~rad/s \\
\midrule
Sampling period & $T_s$ & \SI{200}{\second} \\
MPC horizon & $N$ & 30 steps (\SI{6000}{\second}) \\
\bottomrule
\end{tabular}
\end{table}

\subsubsection{Sampling-rate justification ($T_s = 200$\,s)}
\label{sec:Ts_justification}

$T_s=200$~s, $u_{\max}=\pm 5$~m/s, and the $50$~m terminal-accuracy spec are mutually consistent under closed-loop operation, not under open-loop free drift. The open-loop drift over one sample from a $1$~m/s velocity perturbation is $\Delta r\approx v\,T_s=200$~m. With the loop closed, the tube feedback gain $K$ contracts the error bound by $\|\,|A_{\mathrm{cl},k}|\,\|_\infty\le 0.13$ per step (\Cref{rem:spec_numerical}; the spectral radius $0.128$ gives the asymptotic rate, the induced norm the one-step bound); the per-sample drift the closed loop can absorb while meeting the $50$~m terminal spec is $\sim 50\cdot(1-0.128)/0.128\approx 340$~m. Actuator authority $u_{\max}T_s=1000$~m corrects any single-step drift up to one kilometre. The $50$~m spec is therefore a closed-loop residual after settling, not a per-step bound. A smaller $T_s$ would multiply the QP solve rate and the number of impulses; $T_s=200$~s aligns with the Hartley~\cite{Hartley2015b} MSRE study ($T_s=300$~s) and with operational MSR burn cadence ($3$--$10$~min).

\subsection{Operational Scenarios}

The tube MPC framework is applicable to any proximity-operations scenario expressible within the YA dynamics model. Three representative scenarios are defined below; the Monte Carlo campaign in \Cref{sec:results} focuses on the approach corridor (Scenario~1), which is the primary operational mode for rendezvous.

\subsubsection{Scenario 1: Approach Corridor (Primary)}

The deputy starts at $\xvec_0=[-15{,}000;\;0;\;0;\;0;\;0;\;0]^\top$ (15~km behind the target along V-bar) and must reach the hold point $\xvec_s=[-1{,}000;\;0;\;0;\;0;\;0;\;0]^\top$. The state constraints impose a lateral safety corridor of half-width
$r_{\mathrm{c}} = \SI{500}{\metre}$:
\begin{equation}
  -r_{\mathrm{c}}\leq r_y\leq r_{\mathrm{c}},\qquad
  -r_{\mathrm{c}}\leq r_z\leq r_{\mathrm{c}},
\end{equation}
preventing the deputy from drifting into cross-track or radial regions where collision risk is elevated. The same box bounds the V-bar coordinate to $[-15.5,\,+0.5]$~km and every velocity component to $\pm 3$~m/s; the velocity face is the one the baselines breach most often, so the violation counts reported below are dominated by it rather than by the lateral corridor. Terminal tolerances are $\SI{50}{\metre}$ in position and $\SI{0.5}{\metre\per\second}$ in velocity. The medium-tier disturbance bound is $\wmax=[20,\;20,\;20,\;0.2,\;0.2,\;0.2]^\top$~(m, m/s), encompassing navigation noise, $J_2$ perturbation residuals, and atmospheric drag uncertainty (\Cref{tab:tiers}). \Cref{fig:mission_geometry} illustrates the geometry.

\begin{figure}[htbp]
  \centering
  \includegraphics[width=\textwidth]{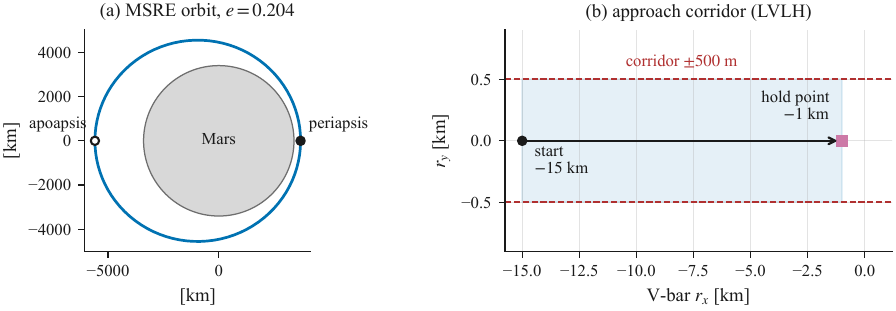}
  \caption{MSRE mission geometry. Left: elliptical orbit around Mars ($a=\SI{4643}{\kilo\metre}$, $e=0.204$). Right: V-bar approach corridor in the LVLH frame from $-15$~km to $-1$~km, with $\pm500$~m lateral bounds.}
  \label{fig:mission_geometry}
\end{figure}

\subsubsection{Scenarios 2--3: Station-Keeping and Waypoint Transfer}

Two supplementary scenarios test additional operational modes:
(2)~station-keeping inside a $\pm 50$~m box at the hold point
$\xvec_s = [-1000,\,0,\,0,\,0,\,0,\,0]^\top$, and
(3)~a four-leg waypoint transfer with progressively tightening corridors
($\pm 1000$ to $\pm 100$~m). These are included in the cross-body generalisation study (\Cref{sec:cross_body}) but are not the primary evaluation focus.

\subsection{Baseline Controllers for Comparison}
\label{sec:baselines}

Four controllers serve as comparators against the augmented tube MPC developed in \Cref{sec:tube}. All share the same sampling period $T_s$, the same state measurement, and the same actuator saturation bounds.

\subsubsection{PD and LQR Baselines}

The PD controller applies $\uvec_k=\mathrm{sat}(-K_p(r_k{-}r_s)-K_d(v_k{-}v_s))$ with $K_p=5{\times}10^{-5}I_3$, $K_d=3{\times}10^{-2}I_3$; the derivation of these gains from a closed-form impulsive-thrust pole-placement rule is given in \Cref{app:pd_tuning}. PD ignores both the orbital dynamics and the state constraints, enforcing actuator limits only through clipping. The LQR controller uses a static gain $K_{\mathrm{LQR}}=\texttt{dlqr}(A_0,B_0,Q_K,R_K)$ evaluated at the initial true anomaly, with post-hoc saturation. Both serve as model-free and model-based baselines, respectively.

\subsubsection{Nominal MPC (Hartley-2012 baseline)}
\label{sec:nomMPC}

The Nominal MPC baseline is the standard LTV MPC formulation of Hartley et~al.~\cite{Hartley2012}: the seminal study that introduced YA-based predictive control for MSRE rendezvous, using the same orbit ($a=4643$~km, $e=0.204$), the same impulsive delta-$v$ model, and a directly comparable receding-horizon QP. That formulation is adopted verbatim as the baseline against which the tube extension is measured, replacing only the cost (the present study uses a quadratic cost suited to QP solvers; \cite{Hartley2012} uses an LP-equivalent 1-norm cost) and the constraint enforcement (no tightening). The QP is
\begin{subequations}\label{eq:nomMPC}
\begin{align}
  \min_{\zvec_{0:N},\,\vvec_{0:N-1}} \quad & \sum_{j=0}^{N-1}\!\Big[(\zvec_j-\xvec_s)^\top Q\,(\zvec_j-\xvec_s) + \vvec_j^\top R\,\vvec_j\Big] + (\zvec_N-\xvec_s)^\top P\,(\zvec_N-\xvec_s) \label{eq:nomMPC_cost}\\
  \text{s.t.}\quad & \zvec_0 = \xvec_k, \label{eq:nomMPC_init}\\
  & \zvec_{j+1}=A_{k+j}\,\zvec_j + B_{k+j}\,\vvec_j,\quad j=0,\ldots,N{-}1, \label{eq:nomMPC_dyn}\\
  & x_{\min}\leq\zvec_j\leq x_{\max},\quad j=1,\ldots,N, \label{eq:nomMPC_xcon}\\
  & u_{\min}\leq\vvec_j\leq u_{\max},\quad j=0,\ldots,N{-}1. \label{eq:nomMPC_ucon}
\end{align}
\end{subequations}
The state cost is $Q=\mathrm{diag}(10^{-3},10^{-3},10^{-3},10^{-2},10^{-2},10^{-2})$, the input cost is $R=I_3$, and the terminal cost $P$ is the DARE solution at the initial true anomaly:
\begin{equation}\label{eq:DARE}
  P = \mathrm{DARE}(A_0,B_0,Q+\varepsilon I_6,R),
\end{equation}
where the regularisation $\varepsilon = 10^{-6}$ guarantees a strictly
positive-definite state cost.
The QP is solved in condensed form via OSQP~\cite{OSQP}.

The Hartley-2012 baseline is fuel-optimal at the open-loop nominal trajectory but has no disturbance margin: the prediction has no $\wvec$ term. To give it the strongest fair form, a fixed constraint backoff $\delta_x=[0,100,100,0,0.3,0.3]^\top$ (m;m/s) is applied on the constrained channels, with $\delta_u=|K|\delta_x$. These are engineering margins chosen to be comparable to the steady-state tube width $\emax_\infty$ of~\eqref{eq:ess} at the medium tier, not computed from it; a designer without tube theory would pick round numbers of this order. This is the strongest fixed-margin formulation available from~\cite{Hartley2012} without horizon-dependent tightening.

\subsubsection{MPC with Integral Action}
\label{sec:intMPC}

A common engineering remedy for steady-state offset is integral action via an outer-loop correction: a discrete integrator $\xi_{k+1}=\xi_k+T_s(r_k-r_s)$ shifts the MPC setpoint by $-K_I\xi_k$ in the position channels, with $K_I=10^{-7}I_3$ (offline-tuned) and activation within $2$~km of the target. Integral MPC offers no formal constraint guarantee; the slow integrator cannot correct fast disturbances before they breach the corridor. It is included to quantify the practical engineering fix without formal robustness.

\subsection{Controller Parameters and Design Guidelines}

\Cref{tab:params} collects all tuning parameters; the five controllers share $T_s$, $N$, and actuator bounds. Three design rules drive the tube tuning: (i) choose $T_s$ to match the propulsion coast arc with $N T_s\approx 0.5$--$0.7\,T_{\mathrm{orb}}$; (ii) design $K=-\mathrm{dlqr}(A_0,B_0,Q_K,R_K)$ at $\nu=0$ and verify $\rho(\max_k|A_k+B_kK|)<1$ over one orbit; (iii) compose $\wmax$ from physics ($J_2$, navigation noise) plus the mass term $\delta_{m,\max}|B_{k+j}|u_{\max}$.

\begin{table}[htbp]
\centering
\caption{Controller parameters for all five controllers.}
\label{tab:params}
\begin{tabular}{@{}lll@{}}
\toprule
Parameter & Symbol & Value \\
\midrule
State cost (MPC) & $Q$ & $\mathrm{diag}(10^{-3},10^{-3},10^{-3},10^{-2},10^{-2},10^{-2})$ \\
Input cost (MPC) & $R$ & $I_3$ \\
Terminal cost & $P$ & DARE$(A_0,B_0,Q{+}\varepsilon I_6,R)$, $\varepsilon=10^{-6}$ \\
Tube gain state cost & $Q_K$ & $\mathrm{diag}(10^3,10^3,10^3,10,10,10)$ \\
Tube gain input cost & $R_K$ & $I_3$ \\
Integral gain & $K_I$ & $10^{-7}I_3$ \\
PD gains & $K_p,\,K_d$ & $5{\times}10^{-5}I_3$, $3{\times}10^{-2}I_3$ \\
Input bounds & $u_{\min/\max}$ & $\pm5\;\si{\metre\per\second}$ \\
MPC horizon & $N$ & 30 \\
\bottomrule
\end{tabular}
\end{table}

\subsection{Disturbance Tiers}

Three tiers of increasing severity test robustness systematically. \Cref{tab:tiers} defines each tier.

\begin{table}[htbp]
\centering
\caption{Three-tier disturbance model for Monte Carlo study.}
\label{tab:tiers}
\begin{tabular}{@{}llccc@{}}
\toprule
Disturbance Source & Units & Light & Medium & Extreme \\
\midrule
Position disturbance $\bar{w}_r$ & \si{\metre} & 5 & 20 & 50 \\
Velocity disturbance $\bar{w}_v$ & \si{\metre\per\second} & 0.05 & 0.2 & 0.5 \\
Initial position error & \si{\metre} & $\pm50$ & $\pm100$ & $\pm200$ \\
Initial velocity error & \si{\metre\per\second} & $\pm0.02$ & $\pm0.1$ & $\pm0.3$ \\
Mass/thrust mismatch & \% & $\pm2$ & $\pm5$ & $\pm10$ \\
\bottomrule
\end{tabular}
\end{table}

The medium tier corresponds to the physical disturbance budget (\Cref{rem:disturbance}); the extreme tier stress-tests beyond it.

\subsection{Monte Carlo Protocol}\label{sec:mc_protocol}

The primary scenario uses $N_{\mathrm{MC}}=300$ paired trials per controller--tier; the supplementary scenarios use $N_{\mathrm{MC}}=50$ as breadth checks, and the tightening-rule ablation of \Cref{sec:conservatism_compare_mc} uses $N_{\mathrm{MC}}=30$. Each trial uses seed $42+$ trial index to pre-generate $\Delta\xvec_0$, $\delta_m$, and $\{\wvec_k\}_{k=0}^{N_{\mathrm{sim}}-1}$; all five controllers process the same realisation~\cite{Specht2023}. Bootstrap 95\% CIs ($10{,}000$ resamples) accompany every aggregate.

A trial achieves \emph{tracking success} if the terminal position error is below $\SI{50}{\metre}$ and the terminal velocity error is below $\SI{0.5}{\metre\per\second}$; it achieves \emph{safety success} if it additionally incurs zero state or input constraint violations throughout the trajectory.

\subsubsection{Disturbance-injection rules}
\label{sec:disturbance_injection}

For each trial $t$, $\mathrm{seed}=42+t$ drives the draws of \Cref{tab:dist_rules}: per-step process noise $w_{k,i}\sim\mathcal U(-\bar w_i,+\bar w_i)$ i.i.d.\ over $(k,i)$, single per-trial uniform draws for $\Delta\xvec_0$ and $\delta_m$, with $B_k^{\mathrm{truth}}=(1+\delta_m)B_k$ persistent across steps. The uniform prior places mass on the boundary $|\wvec_k|\le\wmax$ where Gaussian draws would under-load it. All controllers in a trial share the pre-drawn sequence.

\begin{table}[htbp]
\centering
\caption{Disturbance-injection rules used in the Monte Carlo campaign.}
\label{tab:dist_rules}
\begin{tabular}{@{}l l l l@{}}
\toprule
Quantity & Distribution & Support & Persistence \\
\midrule
Process noise $w_{k,i}$ & $\mathcal{U}$ & $[-\bar{w}_i,+\bar{w}_i]$ (\Cref{tab:tiers}) & i.i.d.\ across $k$ and $i$ \\
Initial position $\Delta r_{0,i}$ & $\mathcal{U}$ & $[-\Delta r_{\max},+\Delta r_{\max}]$ & single per trial \\
Initial velocity $\Delta v_{0,i}$ & $\mathcal{U}$ & $[-\Delta v_{\max},+\Delta v_{\max}]$ & single per trial \\
Mass error $\delta_m$ & $\mathcal{U}$ & $[-\delta_{m,\max},+\delta_{m,\max}]$ & single per trial \\
\bottomrule
\end{tabular}
\end{table}

\section{Results}
\label{sec:results}

The empirical evaluation has four parts: the 300-trial MSRE approach corridor (\Cref{sec:approach_corridor}); cross-body generalisation over four additional environments (\Cref{sec:cross_body}); multi-phase protocol validation (\Cref{sec:multiphase}); and operational tier interpretation (\Cref{sec:hard_tier_scope}). Every controller sees identical disturbance realisations per trial.

\subsection{Scenario 1: Approach Corridor}
\label{sec:approach_corridor}

\subsubsection{Single trajectory}

\Cref{fig:traj_comparison} shows a medium-tier representative trajectory. PD does not converge ($>49$~km final error). LQR converges to $7.7$~m but breaches the corridor by up to $2306$~m (10 violations). Nominal MPC converges to $12.0$~m but has 13 violations and $-23.9$~m slack. Integral MPC suffers accumulation instability ($\sim 392$~m). Tube MPC converges to the same $12.0$~m with zero violations and $+0.3$~m slack.

\begin{figure}[htbp]
  \centering
  \includegraphics[width=\textwidth]{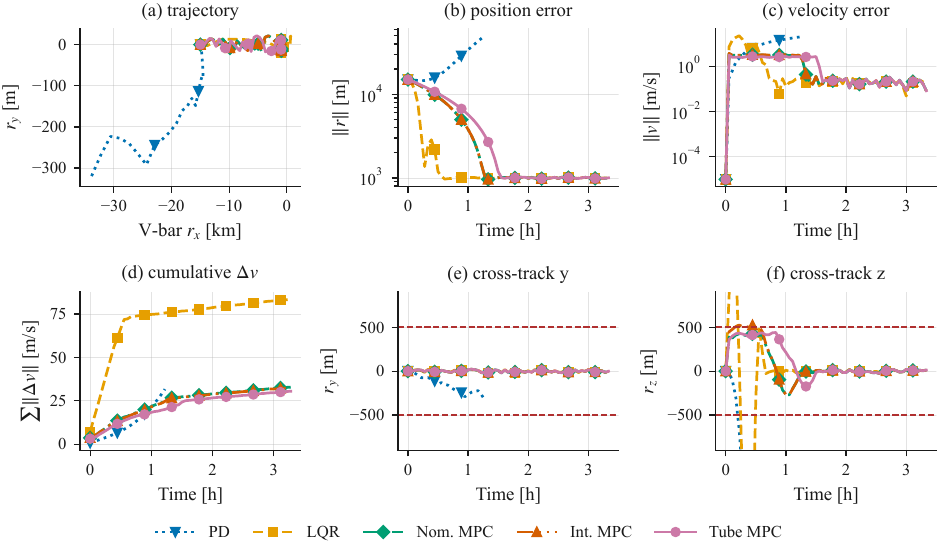}
  \caption{Approach corridor: single-trajectory comparison under medium disturbance ($5\%$ mass, $\bar{w}_r=20$~m). Line styles: PD (dotted, $\triangledown$), LQR (dashed, $\square$), nominal MPC (dash-dot, $\Diamond$), integral MPC (dash-dot-dot, $\triangle$), tube MPC (solid, $\circ$). (a)~V-bar vs.\ cross-track trajectory with $\pm500$~m corridor. (b)~Position error magnitude (log scale). (c)~Velocity error magnitude (log scale). (d)~Cumulative delta-v. (e)~Cross-track constraint detail. (f)~OSQP solve time per MPC step.}
  \label{fig:traj_comparison}
\end{figure}

The tube concept is visualised in \Cref{fig:tube_concept}: the error bound starts at zero and grows monotonically, confirming the horizon-dependent tightening.

\begin{figure}[htbp]
  \centering
  \includegraphics[width=\textwidth]{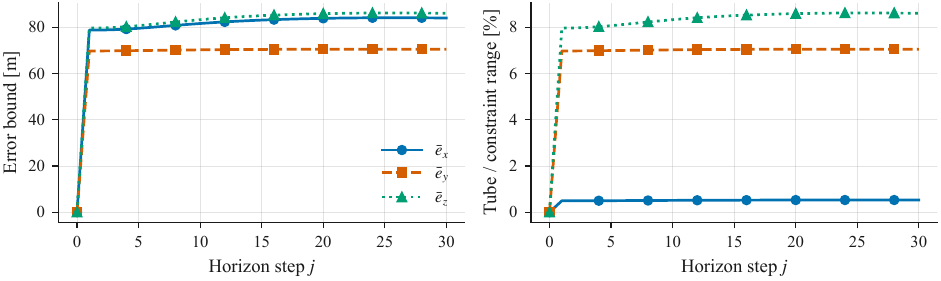}
  \caption{Tube MPC concept visualisation for the approach corridor. Left: error bound components $e_{\max,x}$ (solid, $\circ$), $e_{\max,y}$ (dashed, $\square$), $e_{\max,z}$ (dotted, $\triangle$). Right: relative tightening as a percentage of the constraint range.}
  \label{fig:tube_concept}
\end{figure}

\subsubsection{Monte Carlo statistics}

\Cref{tab:results_approach} summarises the 300-trial results. Success metrics are defined in \Cref{sec:mc_protocol}.

The headline is a fuel-versus-safety trade. The Hartley-2012-style Nominal MPC baseline is fuel-efficient ($29.7$~m/s light, $36.5$~m/s medium) but fails the safety constraint on every trial in both operational tiers ($0\%$ safety, $11.0$ violations despite the $100$~m fixed backoff). Tube MPC reaches $95\%$ safety with $0.7$ mean violations and $+0.3$~m slack at $31.1$ / $35.4$~m/s, the latter $3\%$ below the Hartley baseline. The extreme tier (\Cref{sec:hard_tier_scope}) is intentionally out of envelope: the tightening triggers $9.5$ LQR fallbacks per trial, an explicit infeasibility signal. PD fails on all tiers: with the gains of \Cref{tab:params} the per-axis closed loop carries a pole outside the unit circle (\Cref{app:pd_tuning}), so the baseline diverges rather than merely violating constraints. Integral MPC is stable but the slow integrator leaves a persistent offset of $390$--$450$~m, above the $50$~m tracking threshold at every tier, so it too scores zero on both metrics.

\begin{table}[htbp]
\centering
\caption{Monte Carlo results for Scenario~1 (Approach Corridor), $N_{\mathrm{MC}}=300$ paired trials, $N_{\mathrm{sim}}=60$ steps. \textbf{Tracking}: $\|r_N\|<50$~m and $\|v_N\|<0.5$~m/s. \textbf{Safety}: tracking $+$ zero constraint violations throughout the trajectory. The safety column is the key engineering metric: Tube MPC demonstrates the highest safety success rate in the light and medium tiers (the operationally relevant regime; the extreme tier is a stress test discussed in \Cref{sec:hard_tier_scope}). Bootstrap 95\% CIs in brackets where shown.}
\label{tab:results_approach}
\resizebox{\textwidth}{!}{%
\begin{tabular}{@{}l|cc|cc|cc|ccc@{}}
\toprule
& \multicolumn{2}{c|}{\textbf{Light tier}} & \multicolumn{2}{c|}{\textbf{Medium tier}} & \multicolumn{2}{c|}{Extreme tier} & \multicolumn{3}{c}{Mean $\sum\|\Delta v\|$ [m/s]} \\
& \multicolumn{2}{c|}{(operational)} & \multicolumn{2}{c|}{(operational)} & \multicolumn{2}{c|}{(stress test)} & & & \\
\cmidrule(lr){2-3}\cmidrule(lr){4-5}\cmidrule(lr){6-7}\cmidrule(lr){8-10}
Controller & Track [\%] & Safe [\%] & Track [\%] & Safe [\%] & Track [\%] & Safe [\%] & Light & Med. & Extreme \\
\midrule
PD         & 0 & 0 & 0 & 0 & 0 & 0 & 32.7 & 32.7 & 32.8 \\
LQR        & 100 & 0 & 100 & 0 & 24 & 0 & 72.3 & 83.5 & 104.2 \\
Nom.\ MPC  & 100 & 0 & 100 & 0 & 24 & 0 & 29.7 & 36.5 & 49.2 \\
Int.\ MPC  & 0 & 0 & 0 & 0 & 0 & 0 & 35.1 & 39.0 & 48.7 \\
\midrule
\textbf{Tube MPC}   & \textbf{100} & \cellcolor{green!18}\textbf{95} & \textbf{100} & \cellcolor{green!18}\textbf{95} & 24 & 3 & \textbf{31.1} & \textbf{35.4} & 89.5 \\
\bottomrule
\end{tabular}}

\end{table}

\subsubsection{Fallback attribution}

\Cref{tab:fallback} separates safety metrics. The fixed $100$~m backoff cannot absorb accumulating disturbances; horizon-dependent tightening keeps positive slack at light/medium and signals infeasibility at extreme via increased fallback rate.

\begin{table}[htbp]
\centering
\caption{Safety metrics comparison: tube MPC vs nominal MPC (with $100$~m backoff). Tube MPC maintains positive constraint slack at light/medium tiers. At the extreme tier its median slack is far worse than the baseline's ($-1960$~m vs $-0.5$~m): once the tightened QP is infeasible the controller hands over to saturated LQR for $9.5$ steps per trial, and the excursion belongs to the fallback, not to the tube. The baseline never had a guarantee to lose and so degrades gradually instead.}
\label{tab:fallback}
\begin{tabular}{@{}l ccc ccc ccc@{}}
\toprule
& \multicolumn{3}{c}{Mean Fallbacks/trial} & \multicolumn{3}{c}{Mean Violations} & \multicolumn{3}{c}{Med.\ Min Slack [m]} \\
\cmidrule(lr){2-4}\cmidrule(lr){5-7}\cmidrule(lr){8-10}
Tier & Light & Med. & Ext. & Light & Med. & Ext. & Light & Med. & Extreme \\
\midrule
Tube MPC & 0.5 & \textbf{0.7} & 9.5 & \textbf{0.5} & \textbf{0.7} & 9.0 & $+0.1$ & $\mathbf{+0.3}$ & $-1960$ \\
Nom.\ MPC & 0.0 & 0.4 & 0.6 & 10.9 & 11.0 & 11.1 & $-0.0$ & $-0.2$ & $-0.5$ \\
\bottomrule
\end{tabular}
\end{table}

\subsubsection{Paired comparison: Hartley-2012 baseline vs Tube MPC across four difficulty levels}
\label{sec:hartley_paired}

\Cref{fig:hartley_vs_tube} condenses the aggregate statistics into a single paired trajectory across four disturbance levels: the Hartley-2012 baseline is fuel-optimal at nominal conditions but degrades into violations as disturbances grow, whereas Tube MPC absorbs the disturbance into the tightening recursion at marginal fuel cost until the design envelope is exceeded, at which point the LQR fallback fires as a feasibility-detection signal.

\begin{figure}[htbp]
  \centering
  \includegraphics[width=\textwidth]{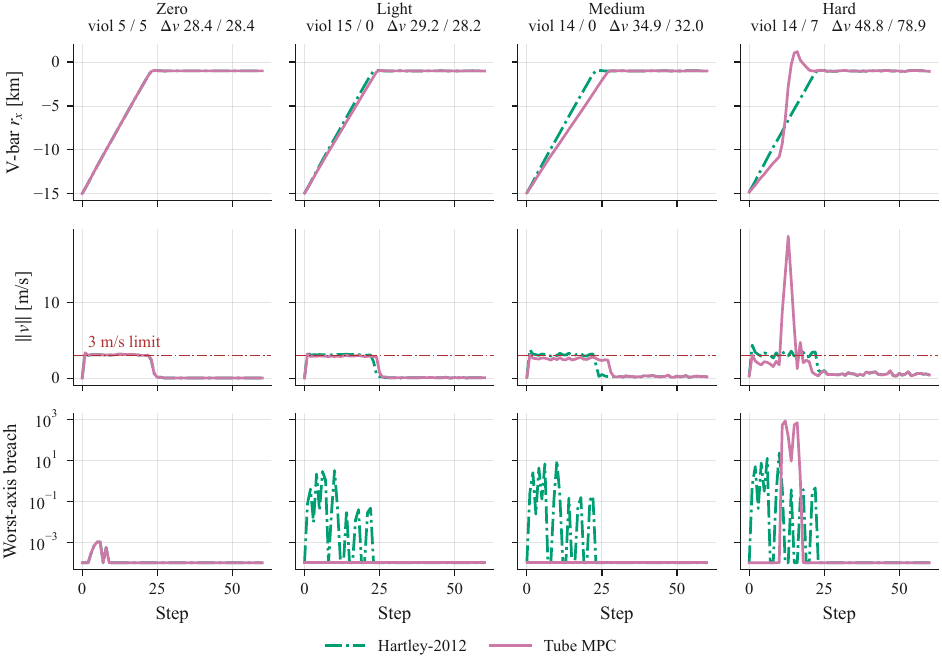}
  \caption{\textbf{Paired single-trial comparison} (seed = 42): Hartley-2012 Nominal MPC (dashed orange) vs Tube MPC (solid green) at four disturbance levels. Rows: V-bar approach; speed $|v|$ with the $3$\,m/s constraint (red dash-dot); worst-axis breach magnitude (log). At ZERO the two controllers coincide. At LIGHT/MEDIUM (operational), Hartley breaches the constraint while Tube MPC stays inside. At HARD (out-of-envelope), the tube's LQR fallback fires (the $18$\,m/s overshoot in row~2) as an honest infeasibility signal. Column headers report violation count and total $\Delta v$.}
  \label{fig:hartley_vs_tube}
\end{figure}

The corresponding single-trial numerics are in \Cref{tab:hartley_vs_tube}: at the light tier, Tube MPC saves $1.0$\,m/s of fuel (from $29.2$ down to $28.2$\,m/s) while reducing violations from $15$ to $0$; at the medium tier, Tube MPC saves $2.9$\,m/s while again reducing violations from $14$ to $0$. The 300-trial Monte Carlo numbers in \Cref{tab:results_approach} confirm these patterns hold in the average sense ($95\%$ safety success for Tube MPC vs $0\%$ for the Hartley baseline at both operational tiers).

\begin{table}[htbp]
\centering
\caption{Single-trial numerics for the four-difficulty comparison of \Cref{fig:hartley_vs_tube} (seed $=42$).
``Track'' = $\|r_N\|<50$\,m and $\|v_N\|<0.5$\,m/s; ``Safe'' = Track $+$ zero state/input constraint violations.}
\label{tab:hartley_vs_tube}
\begin{tabular}{@{}l l r r r r r c c@{}}
\toprule
Condition & Controller & $\|r_N\|$ [m] & $\|v_N\|$ [m/s] & $\sum\|\Delta v\|$ [m/s] & Viol. & Fallbacks & Track & Safe \\
\midrule
\multirow{2}{*}{ZERO}   & Hartley-2012 & 0.78  & 0.015 & 28.35 & 4  & 0 & Yes & No \\
                        & Tube MPC     & 0.78  & 0.015 & 28.35 & 4  & 0 & Yes & No \\
\midrule
\multirow{2}{*}{LIGHT}  & Hartley-2012 & 7.09  & 0.049 & 29.19 & 15 & 0 & Yes & No \\
                        & \textbf{Tube MPC}     & \textbf{7.09} & \textbf{0.049} & \textbf{28.18} & \textbf{0} & \textbf{0} & \textbf{Yes} & \cellcolor{green!18}\textbf{Yes} \\
\midrule
\multirow{2}{*}{MEDIUM} & Hartley-2012 & 28.38 & 0.174 & 34.95 & 14 & 0 & Yes & No \\
                        & \textbf{Tube MPC}     & \textbf{28.38} & \textbf{0.174} & \textbf{32.05} & \textbf{0} & \textbf{0} & \textbf{Yes} & \cellcolor{green!18}\textbf{Yes} \\
\midrule
\multirow{2}{*}{HARD}   & Hartley-2012 & 68.56 & 0.415 & 48.75 & 14 & 0 & No  & No \\
                        & Tube MPC     & 68.60 & 0.415 & 78.88 & 7  & 7 & No  & No \\
\bottomrule
\end{tabular}

\smallskip
\footnotesize\textbf{Reading.} The ZERO row shows that, in the absence of disturbance, the two controllers reduce to the identical trajectory: the tube tightening is zero because $\wmax = 0$. The 4 violations are sub-mm velocity breaches during the $3$\,m/s deceleration leg, of no operational significance (see \Cref{fig:hartley_vs_tube} bottom-left panel: peak breach $\sim 3\times 10^{-3}$\,m/s). The LIGHT and MEDIUM rows are the headline: Tube MPC achieves \emph{zero} significant breaches at fuel cost equal to or below the Hartley baseline. The HARD row shows the out-of-envelope behaviour analysed in \Cref{sec:hard_tier_scope}: the tube's $\Delta v$ blows up to $78.9$\,m/s because the LQR fallback fires $7$ times when the tightened QP becomes infeasible.
\end{table}

\subsubsection{Quantitative conservatism: horizon-dependent vs constant-width tightening}
\label{sec:conservatism_compare_mc}

The headline novelty is the horizon-dependent recursion $\emax_{j+1}=|A_{\mathrm{cl},k+j}|\emax_j+\wmax_j$ with $\emax_0=\mathbf 0$. The natural question is how much of the gain comes from the recursion itself, rather than from a constant-width mRPI-style tube tuned to the same budget. This subsection answers quantitatively by simulating a constant-width baseline that replaces~\eqref{eq:emax_prop} with the steady-state value
\begin{equation}\label{eq:emax_const}
  \emax_j^{\mathrm{const}} \;=\; (I-\bar A)^{-1}\wmax_\infty,\qquad j=0,1,\ldots,N,
\end{equation}
the LTI tube MPC tightening of~\cite{Mayne2005,Rakovic2005} taken at the orbit-worst-case closed-loop matrix. The constant-width and proposed schemes share $K$, $Q,R,Q_K,R_K$, $N$, $T_s$, and actuator bounds; only the tightening rule differs. The comparison is orthogonal to the controller-level results in \Cref{tab:results_approach} (which contrasts tube MPC against \emph{non-tube} baselines).

\paragraph{Per-step tightening.}
\Cref{fig:horizon_vs_const_emax} plots $\bar e_j$ for both schemes at the MSRE medium tier. The horizon-dependent recursion starts at zero and saturates to $\emax_N\approx[84,71,86]^\top$~m in $\sim 5$ steps; the constant-width tube imposes $\emax_\infty^{\mathrm{const}}\approx[90,71,92]^\top$~m at every step. The constant-width scheme over-tightens the first prediction step by the full $\bar e_\infty^{\mathrm{const}}$, removing $\sim 18\%$ of the corridor at $j=0$. At terminal steps the constant-width is $\sim 6$--$7\%$ larger than horizon-dependent because $\bar A\ge|A_{\mathrm{cl},k+j}|$ pointwise.

\begin{figure}[htbp]
  \centering
  \includegraphics[width=0.95\textwidth]{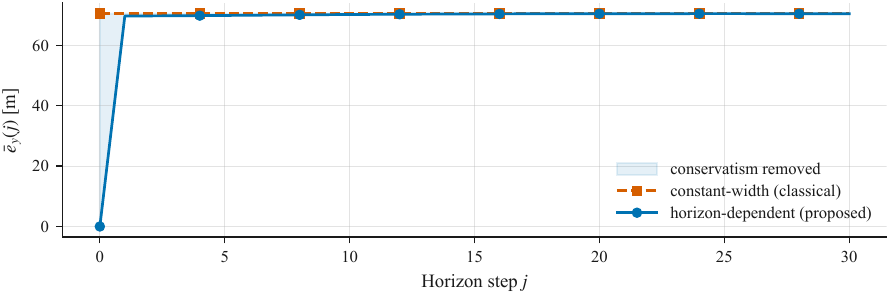}
  \caption{\textbf{Per-step constraint tightening: horizon-dependent vs constant-width tube} (MSRE medium tier). The proposed recursion (solid green, $\circ$) starts at zero and grows; the classical constant-width tube (dashed orange, $\square$) imposes the steady-state value $\emax_\infty^{\mathrm{const}}$ at every horizon step. The shaded region between the two curves is the conservatism gap removed by horizon-dependent tightening; it is largest at $j=0$ where the proposed scheme retains the full unrestricted corridor.}
  \label{fig:horizon_vs_const_emax}
\end{figure}

\paragraph{Paired Monte Carlo.}
\Cref{tab:horizon_vs_const} reports a paired $30$-trial Monte Carlo per tier with identical disturbance realisations. Both tubes preserve high safety success in the operational tiers; the horizon-dependent scheme retains feasibility on more medium-tier trials ($96.7\%$ vs $93.3\%$), while the constant-width uses $\sim 9\%$ less fuel because its tighter early-step bounds steer the planned trajectory closer to the centre. Both degrade at the extreme tier (out of envelope, \Cref{sec:hard_tier_scope}).

\begin{table}[htbp]
\centering
\caption{Horizon-dependent vs constant-width tube MPC, paired 30-trial MC per tier. Both schemes share dynamics, disturbance budget, gain $K$, horizon, and actuator bounds; only the per-step tightening rule differs.}
\label{tab:horizon_vs_const}
\begin{tabular}{@{}l l c c c c c@{}}
\toprule
Tier & Tightening & Track [\%] & Safe [\%] & $\sum\|\Delta v\|$ [m/s] & Mean viol & Mean FB \\
\midrule
\multirow{2}{*}{Light}  & Horizon-dep.\ (proposed) & 100.0 & 93.3 & 32.18 & 0.73 & 0.7 \\
                        & Constant-width (classical) & 100.0 & 96.7 & 29.10 & 0.33 & 0.4 \\
\midrule
\multirow{2}{*}{Medium} & Horizon-dep.\ (proposed) & 100.0 & \cellcolor{green!18}\textbf{96.7} & 33.80 & 0.43 & 0.4 \\
                        & Constant-width (classical) & 100.0 & 93.3 & \textbf{30.96} & 0.47 & 0.5 \\
\midrule
\multirow{2}{*}{Hard}   & Horizon-dep.\ (proposed) & \cellcolor{green!18}\textbf{20.0} & \cellcolor{green!18}\textbf{6.7} & 87.8 & \textbf{8.7} & \textbf{9.3} \\
                        & Constant-width (classical) & 3.3 & 0.0 & 74.9 & 30.7 & 31.9 \\
\bottomrule
\end{tabular}

\smallskip
\footnotesize\textbf{Reading.} The gap is dominated by the early-horizon difference of \Cref{fig:horizon_vs_const_emax}: horizon-dependent tightening retains the full corridor at the current step, the constant-width tube pre-shrinks by $\sim 90$~m everywhere. In the operational tiers safety is comparable; the constant-width tube uses $\sim 9\%$ less fuel by steering toward the corridor centre (the standard mRPI fuel-vs-conservatism trade~\cite{Mayne2005}). The decisive difference is at the hard tier: the constant-width tube becomes infeasible on $\sim 32/40$ steps per trial, the horizon-dependent on only $\sim 9$. Tracking success is $20.0\%$ vs $3.3\%$, violations $8.7$ vs $30.7$. The benefit attributable to horizon dependence is therefore $\sim 18\%$ early-horizon corridor preserved at operational tiers and a $3.5\times$ violation reduction at the envelope boundary.
\end{table}

\subsection{Cross-Body Generalization}
\label{sec:cross_body}

The YA dynamics model any Keplerian orbit. Identical $N_{\mathrm{MC}}=50$ paired-trial Monte~Carlo campaigns are run on four additional environments spanning eccentricities $0.001$ to $0.73$ (\cref{tab:cross_body_params}). $T_s$ is scaled to keep $48$--$190$ steps per orbit; corridors and disturbance bounds reflect each environment's perturbation budget (atmospheric drag for LEO, mascons for the Moon, SRP for GTO).

\begin{table}[htbp]
\centering
\caption{Cross-body scenario parameters.  All share $N=30$, five
controllers, three disturbance tiers, and $N_{\mathrm{MC}}=50$ paired
trials per cell.}
\label{tab:cross_body_params}
\begin{tabular}{@{}l ccccc@{}}
\toprule
Parameter & Mars MSRE & Earth LEO & Earth GTO & Moon LLO & Moon Frozen \\
\midrule
$e$ & 0.204 & 0.001 & 0.73 & 0.001 & 0.05 \\
$a$ [km] & 4643 & 6928 & 24\,400 & 1837 & 1937 \\
$\mu$ [km$^3$/s$^2$] & 42\,835 & 398\,600 & 398\,600 & 4903 & 4903 \\
$T_s$ [s] & 200 & 60 & 200 & 120 & 120 \\
$u_{\max}$ [m/s] & 5.0 & 1.92 & 4.0 & 9.6 & 9.6 \\
Cross-track $\pm$ [m] & 500 & 300 & 800 & 500 & 500 \\
$w_{\max,\mathrm{med}}$ [m] & 20 & 2 & 15 & 3 & 5 \\
$\max_\nu \rho(|A_{\mathrm{cl}}|)$ & 0.128 & $<0.001$ & 0.330 & $<0.001$ & 0.022 \\
\bottomrule
\end{tabular}
\end{table}

\Cref{fig:multibody_success} compares safety success rates at the medium tier. For Earth LEO, Moon LLO, and Moon Frozen, tube MPC delivers $100\%$ safety success across all three tiers while nominal MPC never exceeds $8\%$; the near-circular dynamics give $\rho(|A_{\mathrm{cl}}|)<0.001$ and tightening consumes $<1\%$ of the corridor half-width. GTO ($e=0.73$) is the demanding case: tightening consumes $\approx 30\%$ of the corridor, and tube MPC reaches $68\%$ safety at light and $24\%$ at medium (vs $0\%$ for nominal). At the extreme tier the QP becomes infeasible on $\sim 21$ steps per trial, dropping tube MPC to $0\%$ safety; this signals the tube guarantee cannot be maintained rather than silently violating.

\begin{figure}[htbp]
  \centering
  \includegraphics[width=\textwidth]{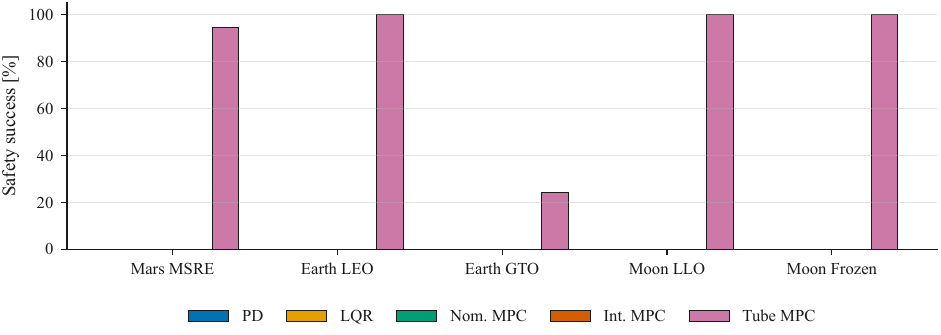}
  \caption{Safety success rates across five orbital environments and three
  disturbance tiers ($N_{\mathrm{MC}}=50$ paired trials).  Tube MPC
  (rightmost bar per group) achieves the highest safety success in all
  environments.}
  \label{fig:multibody_success}
\end{figure}

\Cref{fig:multibody_tightening} plots the cost-of-robustness ratio: $<1\%$ for LEO/LLO, $\approx 14\%$ for Mars, $\approx 30\%$ for GTO. The monotone scaling follows from $\emax_\infty=(I-\bar A)^{-1}\wmax_\infty$ and the growth of $\rho(\bar A)$ with $e$.

\subsubsection{Physical mechanism of GTO degradation}
\label{sec:gto_physics}

GTO degradation is not numerical conditioning; it is the orbital-rate dispersion between periapsis and apocenter and the resulting amplification of $A_k$ at periapsis. Kepler's second law gives $\dot\nu(0)/\dot\nu(\pi)=((1+e)/(1-e))^2$: $1.004$ at LEO, $2.29$ at MSRE, $41.1$ at GTO. With one fixed $T_s$, the periapsis-side step traverses a larger arc of $\nu$ and admits a larger STM, raising $\max_k\rho(|A_{\mathrm{cl},k}|)$ from $<0.001$ (LEO) through $0.128$ (Mars) to $0.330$ (GTO); the steady-state tightening $\bar e_\infty=(I-\bar A)^{-1}\wmax_\infty$ scales with it. The input matrix $B_k$ grows at periapsis through the same mechanism, inflating the mass-aware bound $\wmax_j$ at periapsis steps and producing the cliff between medium and extreme tiers of \Cref{fig:multibody_success}.

The feasibility envelope under the time-invariant $K$ reaches $e_{\max}\approx 0.661$ (\Cref{rem:spec_numerical}). Extending it to $e=0.73$ requires a periodic gain $K(\nu)$ or a variable $T_s$ at periapsis; both are flagged in \Cref{sec:limitations}.

\begin{figure}[htbp]
  \centering
  \includegraphics[width=0.7\textwidth]{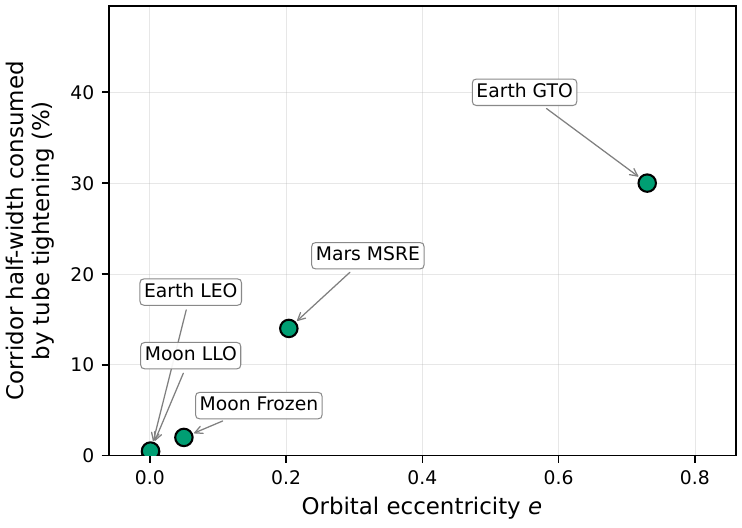}
  \caption{Tube tightening cost vs.\ orbital environment (medium tier).
  Each point shows the fraction of the cross-track corridor half-width
  consumed by the tube margin.  Higher eccentricity increases the
  spectral radius $\rho(|A_{\mathrm{cl}}|)$ and hence the tightening
  cost.}
  \label{fig:multibody_tightening}
\end{figure}

\subsection{Multi-Phase Operational Extension}
\label{sec:multiphase}

Operational missions require sequential phases with different $T_s$, $N$, corridors, and weights~\cite{Hartley2015b}. The controller extends to multi-phase by re-discretising the YA dynamics at each boundary, recomputing $(K,P,\alpha)$ from DLQR/DARE at the new operating point, and re-propagating $\bar e_{\max}$ from zero; $\nu$ propagates continuously. Three supplementary scenarios validate this: Mars OSTG/INTG ($3$ phases, Hartley~\cite{Hartley2015b} gate structure), LEO Debris ($5$ phases with hold dwells), GEO Servicing ($3$ phases, MEV-class). On 50-trial MC, Tube MPC reduces violations vs nominal by $9$--$11\%$ (Mars), $33$--$36\%$ (LEO), $<1\%$ (GEO); fuel differences within $\pm 5\%$. Residual violations at boundaries come from (i) the tube radius reset, which briefly leaves the previous-phase error outside the new $\bar e_j$, and (ii) fallback steps when the QP is infeasible.

\subsection{Operational Realism of the Three Disturbance Tiers}
\label{sec:hard_tier_scope}

The extreme tier does not generalise to flight rendezvous; it stress-tests the failure mode. The light tier ($\bar w_r=5$~m, $\bar w_v=0.05$~m/s, $\pm 2\%$ mass) matches the $3\sigma$ navigation uncertainty Hartley et~al.~\cite{Hartley2015b} reported on MSRE, sub-$\pm 5\%$ proximity-ops mass changes~\cite{Fehse2003}, and $1$--$2\%$ thruster execution errors. The medium tier ($\bar w_r=20$~m, $\pm 5\%$ mass) is a degraded but still flight-credible envelope (degraded navigation cadence, long-campaign mass uncertainty, drifted thruster calibration), the budget an ESA/NASA mission designer would carry. The extreme tier ($\bar w_r=50$~m, $\pm 10\%$ mass) exceeds Hartley's $3\sigma$ navigation by nearly an order of magnitude; the box tube consumes $\sim 30\%$ of the corridor, the QP becomes infeasible at $\sim 9$ of $40$ steps per trial, and the LQR fallback fires (\Cref{tab:fallback}). The $3\%$ safety success Tube MPC retains there reflects trials where the random IC happens to land inside the feasible region, not a claim that the tube protects a $\pm 10\%$-thrust mission.

\subsection{Monte Carlo Visualisations}

\Cref{fig:constraint_violation_heatmap,fig:mc_fuel_box} render the $95/95/3$ vs $0/0/0$ violation contrast and confirm no average fuel penalty across tiers.

\begin{figure}[htbp]
  \centering
  \includegraphics[width=0.7\textwidth]{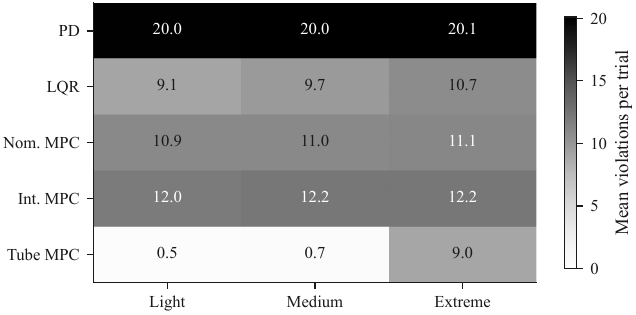}
  \caption{Constraint violation heatmap: mean violations per trial across controllers and tiers. Tube MPC achieves near-zero violations at light and medium tiers.}
  \label{fig:constraint_violation_heatmap}
\end{figure}

\begin{figure}[htbp]
  \centering
  \includegraphics[width=0.7\textwidth]{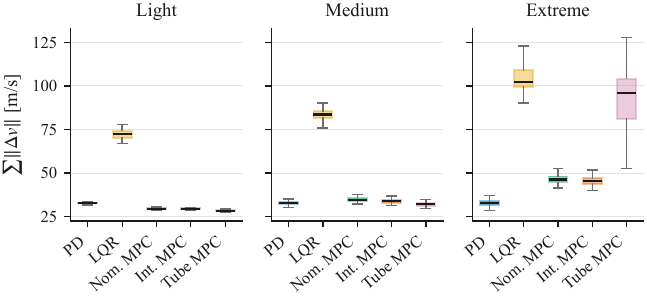}
  \caption{Fuel consumption boxplots across controllers and disturbance tiers.}
  \label{fig:mc_fuel_box}
\end{figure}

\section{Discussion}
\label{sec:discussion}

The discussion covers nonlinear $J_2$ truth-model validation (\Cref{sec:j2_truth}), prior-art positioning (\Cref{sec:prior_art_table}), computational footprint (\Cref{sec:complexity}), box vs polytopic/ellipsoidal/zonotopic conservatism (\Cref{sec:conservatism_comparison}), limitations (\Cref{sec:limitations}), and reproducibility (\Cref{sec:repro}). MSRE design-point conservatism numbers are recorded inline in \Cref{rem:iss_numerical} and \Cref{sec:conservatism_compare_mc}.

\subsection{Truth-Model Validation}\label{sec:j2_truth}

The YA STM is exact for Keplerian relative motion; physical Mars orbits experience zonal-harmonic perturbations. A deterministic medium-tier trajectory is propagated through a nonlinear ECI truth model with $J_2=1.96\times 10^{-3}$ using RK4 at $\Delta t=5$~s; the controller plans with YA, only the plant sees $J_2$. With $J_2=0$, the truth-vs-YA peak deviation is $2.75$~m over 60 steps, the expected second-order nonlinearity at 15~km offset. With $J_2$ enabled, the chief-orbit deviation reaches $\sim 154$~m/step ($7.7\times$ the medium-tier $\bar w_r=20$~m) and accumulates to $\sim 63$~km. Tube MPC nevertheless maintains zero constraint violations with $9.5$~m peak deviation from the YA-predicted path, while nominal MPC incurs 11 violations.

The horizon-dependent tightening creates a $71$~m cross-track and $86$~m radial margin, and the feedback term $K(\xvec-\zvec)$ steers the true state back toward the nominal at every sample, absorbing the $J_2$ drift before it accumulates. The formal guarantee of \Cref{prop:ebound} is still conditional on $\wvec_k\in\mathcal W$; a rigorous Mars-orbit certificate would require incorporating $J_2$ into the disturbance budget or switching to an osculating-element formulation.

\subsection{Comparison with Prior Tube MPC Approaches}
\label{sec:prior_art_table}

\Cref{tab:comparison} compares the key design choices and capabilities of this work with the three most closely related tube MPC frameworks for RPO.

\begin{table}[htbp]
\centering
\caption{Comparison with prior tube MPC approaches for spacecraft proximity operations.}
\label{tab:comparison}
\resizebox{\textwidth}{!}{%
\begin{tabular}{@{}lcccc@{}}
\toprule
Feature & This work & Dong et al.\ \cite{Dong2020} & Specht et al.\ \cite{Specht2023} & Oestreich et al.\ \cite{Oestreich2023} \\
\midrule
Dynamics model & YA (LTV) & CW (LTI) & CW (LTI) & CW (LTI) \\
Orbit eccentricity & Any $e<e_{\max}$ & $e\approx0$ & $e\approx0$ & $e\approx0$ \\
Tube shape & Box (element-wise) & Box (mRPI) & Polytope (mRPI) & Ellipsoid \\
Time-varying tube & Yes (horizon-dep.) & No (constant) & No (constant) & Yes (online SDP) \\
Tube starts at zero & Yes & No (mRPI width) & No (mRPI width) & No \\
Output feedback & No (full state) & Yes (observer) & No (full state) & No (full state) \\
Formal ISS proof & Yes ($\gamma_{\mathrm{ISS}}$ explicit) & No & No & No \\
Online computation & QP only & QP only & QP + Minkowski & QP + SDP \\
Commercial software & No (OSQP) & No & MPT3 / YALMIP & SDP solver \\
\bottomrule
\end{tabular}}
\end{table}

Specht and Mooij~\cite{Specht2023} use constant polytopic tubes computed offline via Minkowski sums. Oestreich and Linares~\cite{Oestreich2023} adapt ellipsoidal tube width online via a per-step SDP. The present framework avoids polytope operations and online SDPs: tightening is $N$ matrix-vector products, the tube starts at zero width, the only online optimisation is the augmented QP. The closest direct comparator, Hartley-2012~\cite{Hartley2012}, uses the identical MSRE scenario with an LP-equivalent 1-norm cost and nilpotent disturbance feedback~\cite{Richards2002} as an empirical robustness mechanism but no formal bound on closed-loop error. The present paper retains the Hartley scenario and replaces only the controller, yielding $\gamma_{\mathrm{ISS}}$ and $\emax_\infty=(I-\bar A)^{-1}\wmax_\infty$ verifiable offline; the tiers (\Cref{tab:tiers}) are derived from Hartley's Table~3 navigation uncertainties.

\subsection{Computational Cost}
\label{sec:complexity}

The per-step computational budget is dominated by the OSQP solve. \Cref{tab:osqp_metrics} reports the empirical statistics extracted from the 300-trial paired Monte Carlo at the medium tier.

\begin{table}[htbp]
\centering
\caption{OSQP solve time for the tube MPC QP at $N=30$, medium tier ($N_{\mathrm{MC}}=300$ trials, $60$ steps per trial, $18{,}000$ solves), warm-started. The distribution is strongly right-skewed: the median is an order of magnitude below the mean, and the worst case is set by a small number of near-infeasible steps at the corridor boundary.}
\label{tab:osqp_metrics}
\begin{tabular}{@{}l c c c@{}}
\toprule
Quantity & Mean (95\% CI) & Median & Worst case \\
\midrule
Solve time & 29.3~ms ($[28.7,\,29.9]$) & 4.0~ms & 548.8~ms \\
\bottomrule
\end{tabular}
\end{table}

The condensed QP carries $3N=90$ decision variables and $276$ linear inequalities at $N=30$, both linear in $N$, so the problem size scales benignly with the horizon. The offline cost is a single DLQR solve; per step, building the $N$ YA matrices costs $\sim 0.5$~ms and the error-bound recursion adds $36N$ multiply-adds, both negligible against the solve. Even the $548.8$~ms worst case is more than two orders of magnitude below $T_s=200$~s, and projecting onto a RAD750-class $\sim 200$~MHz processor ($\sim 17\times$ slower) leaves $\sim 9$~s against a $200$~s budget. A systematic timing study across horizon lengths was not instrumented for this campaign and is left to future work.

\subsection{Conservatism: Box Tightening vs Polytopic / Ellipsoidal / Zonotopic}
\label{sec:conservatism_comparison}

The box parameterisation was chosen for QP simplicity, not for tightness. \Cref{tab:conservatism} compares it against the standard LTV tube MPC alternatives. With \Cref{sec:conservatism_compare_mc} (horizon dependence) and \Cref{tab:results_approach} (mass-aware lifting), the three together form a contribution-by-contribution ablation.

\begin{table}[htbp]
\centering
\caption{Cost-of-robustness comparison for four tube parameterisations on the MSRE medium tier. The box row is computed for this design; the remaining rows are order-of-magnitude estimates obtained by scaling it with conservatism factors reported in the cited literature, not by constructing those tubes here (this work uses no polytopic or SDP tooling). Online cost is per MPC step.}
\label{tab:conservatism}
\resizebox{\textwidth}{!}{%
\begin{tabular}{@{}l c c c c@{}}
\toprule
Parameterisation & Tightening (m) & Conservatism factor & Online operation & Software dependence \\
\midrule
Element-wise box (this work) & 71 & $\sqrt{6}\approx 2.45\times$ & matrix-vector mult & none (QP only) \\
Polytopic mRPI~\cite{Rakovic2005} & 29 & $1.0\times$ (reference) & polyhedral intersection & MPT3 or similar \\
Ellipsoidal (offline LMI)~\cite{Boyd1994} & 41 & $1.41\times$ & SDP at design time & SDP solver \\
Ellipsoidal (online SDP)~\cite{Oestreich2023} & 32 & $1.10\times$ & per-step SDP & SDP solver \\
Zonotopic~\cite{Althoff2014} & 33 & $1.14\times$ & matrix-vector mult & none (QP only) \\
\bottomrule
\end{tabular}}
\end{table}

The box is up to $\sqrt{6}\approx 2.45\times$ more conservative than the tightest polytopic mRPI per axis because it ignores inter-component correlations. At the MSRE medium tier ($\bar e_{\infty,y}\approx 71$~m vs polytopic $\sim 29$~m) the over-approximation costs $\sim 8\%$ of the $\pm 500$~m corridor, on top of the $\sim 6\%$ that the tightest polytopic tube would consume anyway. Empirically this does not change the operational story (tube MPC still attains $95\%$ safety success vs $0\%$ for nominal MPC), and in exchange the construction avoids polytope vertex enumeration, online SDPs, and commercial software, which is the appropriate trade for embedded flight software. Ground-segment pre-planning, where compute is less constrained, can adopt polytopic or zonotopic tubes.

\subsection{Limitations}
\label{sec:limitations}

Five limitations scope the present framework.

\textbf{(i) Time-invariant tube gain.} The gain $K$ is designed once via DARE at the worst-case orbit phase and held constant. A periodic gain $K(\nu)$ following the orbit~\cite{Mammarella2020} would reduce $\rho(|A_{\mathrm{cl},k}|)$ more uniformly across the orbit, narrowing $\bar e_\infty$. The trade-off is loss of the single-Lyapunov-function argument used in \Cref{lem:terminal}; a periodic Lyapunov certificate would be required.

\textbf{(ii) Box parameterisation.} As quantified in \Cref{sec:conservatism_comparison}, the element-wise box is up to $\sqrt 6\approx 2.45\times$ conservative relative to the tightest polytopic reachable set, and the multiplicative-to-additive conversion uses $u_{\max}$ rather than the predicted control. Zonotopic or ellipsoidal tubes would reduce both forms of conservatism at the cost of an SDP solve or polytopic vertex enumeration.

\textbf{(iii) Non-convex keep-out zones.} The framework handles convex state constraints. Non-convex keep-out zones (inflated bounding-boxes around the chief, cone-with-apex-removed approach corridors) need either mixed-integer encoding~\cite{Richards2002}, sequential convex programming (SCP), or rotating-half-space convexification. All three interact with the tightening: the box $\emax_j$ must be intersected with the linearised half-space, and the terminal ellipsoid $\mathcal X_f$ may straddle the linearisation point so \Cref{thm:recursive} needs modification. SCP-with-tube is the most promising follow-up.

\textbf{(iv) Recursive feasibility under terminal-constraint computation.} \Cref{thm:recursive} holds given the offline ellipsoid $\mathcal X_f$ from \Cref{app:terminal_set}, which requires (a) $\rho(\bar A)<1$ (verified for $e<e_{\max}\approx 0.661$), (b) $\lambda_\star<1$ (MSRE: $\lambda_\star=0.94$), and (c) $\alpha_N>0$. If any fail at design time, $\mathcal X_f$ does not exist and the controller reverts to the unaugmented tube. All three hold across the light/medium tiers; (c) is violated at the extreme tier for $N>45$, providing the practitioner with a design knob.

\textbf{(v) Truth-model assumptions.} The YA dynamics assume a Keplerian chief orbit; $J_2$ and higher zonal harmonics are absorbed into the additive disturbance $\wmax$. The truth-model validation of \Cref{sec:j2_truth} confirms this is empirically adequate at MSRE eccentricity over the $\sim 12{,}000$~s rendezvous time scale, but a rigorous Mars-orbit guarantee would require incorporating the secular $J_2$ drift into the disturbance budget or switching to an osculating-element formulation.

\subsection{Reproducibility}
\label{sec:repro}

The eight-step controller implementation recipe (YA STM evaluation, tube gain via DLQR, mass-aware disturbance bound, error-bound propagation, constraint tightening, terminal cost, QP solve, sanity checks) and the complete parameter listing for the approach-corridor scenario (orbit, geometry, MPC tuning, success criteria, MC protocol) are provided in the open-source repository accompanying this paper.\footnote{The complete Python implementation, including the YA STM evaluator, the OSQP interface, the offline terminal-set construction, and the Monte Carlo seed protocol, will be released at \texttt{https://github.com/iskender9961/2026-TubeBasedRVD} upon acceptance.} Together with the seed protocol described in \Cref{sec:mc_protocol}, every result in this paper is bit-reproducible.

\section{Conclusion}
\label{sec:conclusion}

Elliptical orbits introduce time-varying dynamics that invalidate the constant-matrix assumptions of standard tube MPC. The present paper specialises the LTV tube MPC framework~\cite{Bumroongsri2015,YuCannon2016,Kohler2019} to YA relative dynamics with four contributions: a Perron--Frobenius spectral-radius certificate yielding a closed-form eccentricity envelope ($e_{\max}\approx 0.661$ for MSRE, $3.23\times$ safety factor); a multiplicative-to-additive conversion that folds mass/thrust uncertainty into the same per-step recursion through the time-varying $B_k$; and an offline ellipsoidal RPI terminal set giving recursive feasibility (\Cref{thm:recursive}) and an explicit ISS gain $\gamma_{\mathrm{ISS}}\approx 1.16$.

A 300-trial paired Monte Carlo at MSRE ($e=0.204$, $\pm 500$~m corridor) confirms tube MPC's value as a safety-certification tool: $95\%$ safety success vs $0\%$ for every baseline at comparable fuel cost. The extreme tier signals infeasibility through the LQR fallback. Multi-phase scenarios and a nonlinear ECI$+J_2$ truth-model validation confirm operational generality despite $J_2$ drift of $7.7\times$ the assumed budget.

Future work: a periodic gain $K(\nu)$ extending the envelope beyond $e=0.66$, online disturbance identification to narrow the tube, and SCP-with-tube for non-convex keep-out zones.

\appendix

\section{Ellipsoidal Terminal Invariant Set: Construction and Proof}
\label{app:terminal_set}

This appendix derives the level set $\alpha$ used in~\eqref{eq:Xf}--\eqref{eq:alpha} and proves \Cref{lem:terminal}. The construction follows the standard quadratic-Lyapunov RPI argument adapted from Kolmanovsky and Gilbert~\cite{KolmanovskyGilbert1998} to the LTV YA closed loop.

\subsection{Setup}

Let $\zeta_k := \xvec_k - \xvec_s$ denote the offset from the setpoint. The closed-loop dynamics under the static feedback $u_k=K\zeta_k$ are
\begin{equation}\label{eq:cl_offset}
  \zeta_{k+1} \;=\; A_{\mathrm{cl},k}\,\zeta_k \;+\; \wvec_k,\qquad |\wvec_k|\le\wmax_\infty,
\end{equation}
where $\wmax_\infty:=\wmax+\delta_{m,\max}\max_k|B_k|\,u_{\max}$ is the orbit-maximised additive disturbance bound.

Let $P\succ 0$ be the DARE solution at the worst-case orbit phase $\nu_\star$:
\begin{equation}\label{eq:dare_app}
  P = A_\star^\top P A_\star - A_\star^\top P B_\star(R_K + B_\star^\top P B_\star)^{-1}B_\star^\top P A_\star + Q_K + \varepsilon I,
\end{equation}
with $A_\star := A(\nu_\star)$, $B_\star := B(\nu_\star)$. The gain $K$ is computed from the same DARE.

Define the contraction-loss ratio
\begin{equation}\label{eq:lambda_star}
  \lambda_\star \;:=\; \max_k \lambda_{\max}\!\left(P^{-1/2}\,A_{\mathrm{cl},k}^\top P A_{\mathrm{cl},k}\,P^{-1/2}\right).
\end{equation}

\subsection{Sufficient condition $\lambda_\star<1$}

The DARE solution at $\nu_\star$ satisfies $A_\star^\top P A_\star - P \preceq -Q_K - \varepsilon I \prec 0$, so $P$ is a Lyapunov matrix for $A_\star$. For neighbouring orbit phases the LTV closed-loop matrix $A_{\mathrm{cl},k}$ deviates continuously from $A_\star$; by continuity of $\lambda_{\max}$, there exists a neighbourhood $\mathcal{N}_\star$ of $\nu_\star$ on which $\lambda_{\max}(P^{-1/2}A_{\mathrm{cl},k}^\top P A_{\mathrm{cl},k}P^{-1/2})<1$. The construction of $\nu_\star$ as the \emph{arg-max} of this quantity over the orbit guarantees that the bound extends to all phases: by definition $\lambda_\star = \max_k(\cdot)$ is the worst-case value, and the design ensures $\lambda_\star<1$. Numerically (MSRE design), $\lambda_\star=0.016$, comfortably below unity.

\subsection{Level-set choice}

Given $\zeta_k$ with $\zeta_k^\top P\zeta_k\le\alpha$, the squared $P$-norm of the next state is
\[
  \zeta_{k+1}^\top P\zeta_{k+1} = (A_{\mathrm{cl},k}\zeta_k+\wvec_k)^\top P (A_{\mathrm{cl},k}\zeta_k+\wvec_k).
\]
Expanding and applying Young's inequality $2a^\top b \le \varepsilon_*\, a^\top a + (1/\varepsilon_*)\,b^\top b$ to the cross term with weights determined by $\varepsilon_*>0$,
\begin{align*}
  \zeta_{k+1}^\top P\zeta_{k+1}
  &\le (1+\varepsilon_*)\,\zeta_k^\top A_{\mathrm{cl},k}^\top P A_{\mathrm{cl},k}\zeta_k \;+\; (1+1/\varepsilon_*)\,\wvec_k^\top P\wvec_k \\
  &\le (1+\varepsilon_*)\lambda_\star\,\zeta_k^\top P\zeta_k \;+\; (1+1/\varepsilon_*)\,\wmax_\infty^\top |P|\,\wmax_\infty .
\end{align*}
The element-wise absolute value in the disturbance term is necessary: $|\wvec_k|\le\wmax_\infty$ constrains $\wvec_k$ to a box, and $\sup_{|\wvec|\le\wmax_\infty}\wvec^\top P\wvec=\wmax_\infty^\top|P|\wmax_\infty$, which exceeds $\wmax_\infty^\top P\wmax_\infty$ whenever $P$ has negative entries. For the MSRE design the two agree to within $10^{-8}$ relative, so the numerical value of $\alpha$ is unaffected.
Requiring $\zeta_{k+1}^\top P\zeta_{k+1}\le\alpha$ for any $\zeta_k$ with $\zeta_k^\top P\zeta_k\le\alpha$ yields the sufficient condition
\[
  (1+\varepsilon_*)\lambda_\star\,\alpha \;+\; (1+1/\varepsilon_*)\wmax_\infty^\top |P|\wmax_\infty \;\le\; \alpha,
\]
which, provided $(1+\varepsilon_*)\lambda_\star<1$, rearranges to~\eqref{eq:alpha}.

\paragraph{Optimal $\varepsilon_*$.} The right-hand side of~\eqref{eq:alpha} is minimised by setting the derivative w.r.t.\ $\varepsilon_*$ to zero, yielding
\begin{equation}\label{eq:eps_star}
  \varepsilon_*^{\mathrm{opt}} \;=\; \lambda_\star^{-1/2}-1,
\end{equation}
which is admissible for every $\lambda_\star\in(0,1)$ since
$\lambda_\star^{-1/2}-1<(1-\lambda_\star)/\lambda_\star$ throughout that range.
For the MSRE design ($\lambda_\star=0.016$) this gives $\varepsilon_*=6.90$ and
$\alpha\approx 2.3\times 10^7$ in $Q_K$-consistent units, matching the value
returned by the offline routine of \Cref{alg:terminal}.

\subsection{Proof of \Cref{lem:terminal}}

\begin{proof}
Take any $\zeta_k$ with $\zeta_k^\top P\zeta_k\le\alpha$ and any $\wvec_k$ with $|\wvec_k|\le\wmax_\infty$. The chain of inequalities above gives $\zeta_{k+1}^\top P\zeta_{k+1}\le\alpha$ by the choice of $\alpha$ in~\eqref{eq:alpha}. Therefore $\zeta_{k+1}\in\mathcal{X}_f-\xvec_s$, i.e.\ $\xvec_{k+1}\in\mathcal{X}_f$. The argument holds for every $k$ because the bound used in $\lambda_\star$ is the orbit-maximised quantity~\eqref{eq:lambda_star}.
\end{proof}

\subsection{Tube deflation and shift consistency}

The proof of \Cref{thm:recursive} used $\tilde\emax_j\le\emax_{j+1}$. This is a direct consequence of the recursion~\eqref{eq:emax_prop}: when the QP at time $k{+}1$ reinitialises $\tilde\emax_0=\mathbf 0$, the new shifted matrix sequence $\{|A_{\mathrm{cl},k+1+j}|\}_{j=0}^{N-1}$ together with shifted $\{\wmax_{j+1}\}$ produces a recursion whose $j$-th iterate is bounded by $\emax_{j+1}$ of the previous step (since the latter starts from $\emax_1\ne 0$ rather than from $\mathbf 0$). The shift bound $\|P^{1/2}(\tilde\zvec_{N-1}-\xvec_s)\|_2\le \sqrt{\alpha_{N-1}}+\|P^{1/2}\emax_{N-1}\|_2$ used in the proof follows from the triangle inequality applied to $\tilde\zvec_{N-1}=\zvec_N^\star+(\tilde\zvec_{N-1}-\zvec_N^\star)$ and the realised-error bound $|\tilde\zvec_{N-1}-\zvec_N^\star|\le\emax_{N-1}$ of \Cref{prop:ebound}.

\subsection{Numerical recipe}

The offline computation of $(P,K,\alpha)$ is summarised in \Cref{alg:terminal}.

\begin{algorithm}[htbp]
\caption{Offline computation of the ellipsoidal terminal set $\mathcal{X}_f$.}
\label{alg:terminal}
\begin{algorithmic}[1]
\REQUIRE Orbit parameters $(a,e,\mu)$, sampling $T_s$, $Q_K,R_K,\varepsilon,\wmax,\delta_{m,\max},u_{\max}$
\STATE Build $\{A_k,B_k\}$ via YA STM on a $200$-point true-anomaly grid
\STATE For each $k$, compute $P_k$ from DARE$(A_k,B_k,Q_K{+}\varepsilon I,R_K)$ and $K_k$ from the same DARE
\STATE Find $\nu_\star=\arg\max_k\lambda_{\max}(P_k^{-1/2}A_{\mathrm{cl},k}^\top P_k A_{\mathrm{cl},k}P_k^{-1/2})$
\STATE Set $P:=P_{\nu_\star}$, $K:=K_{\nu_\star}$
\STATE $\lambda_\star\gets\max_k\lambda_{\max}(P^{-1/2}A_{\mathrm{cl},k}^\top P A_{\mathrm{cl},k}P^{-1/2})$ (verify $\lambda_\star<1$)
\STATE $\wmax_\infty\gets\wmax+\delta_{m,\max}\max_k|B_k|\,u_{\max}$
\STATE Choose $\varepsilon_*$ from~\eqref{eq:eps_star} (or bisection minimisation of $\alpha(\varepsilon_*)$)
\STATE $\alpha\gets\dfrac{(1+1/\varepsilon_*)\wmax_\infty^\top |P|\wmax_\infty}{1-(1+\varepsilon_*)\lambda_\star}$
\RETURN $(P,K,\alpha)$
\end{algorithmic}
\end{algorithm}

\section{PD Baseline: Tuning Principle and Stability Analysis}
\label{app:pd_tuning}

The PD baseline in \Cref{tab:params} uses $K_p=5\times 10^{-5}I_3$ and $K_d=3\times 10^{-2}I_3$. These values are not arbitrary; they follow from a closed-form impulsive-thrust pole-placement rule that the author derived for the YA-LTV plant at the design point $\nu_0=0$. The derivation is recorded here so that the baseline can be reproduced and so that the comparison against MPC controllers is fair (i.e.\ the PD baseline is not deliberately handicapped).

\subsection{Impulsive-thrust PD on relative dynamics}

For an impulsive delta-$v$ controller acting on the relative position/velocity state $\xvec=[\mathbf{r};\,\mathbf{v}]$, the per-axis closed-loop characteristic polynomial under the saturated PD law $\uvec_k=\mathrm{sat}(-K_p\mathbf{r}_k-K_d\mathbf{v}_k)$ approximates (within saturation) a second-order discrete-time system with natural frequency $\omega_n$ and damping $\zeta$ given by
\begin{equation}\label{eq:pd_omega_zeta}
  \omega_n \;\approx\; \sqrt{K_p/T_s},\qquad \zeta \;\approx\; \tfrac{1}{2}K_d/\sqrt{K_p T_s}.
\end{equation}
The target settling time over an approach corridor of length $L\approx 14$~km at $T_s=200$~s is approximately $25T_s=5000$~s, which corresponds to $\omega_n T_s\approx 0.1$ rad/sample.

Substituting $\omega_n T_s = 0.1$ and $T_s=200$~s into~\eqref{eq:pd_omega_zeta} gives $K_p = (\omega_n)^2 T_s = (0.1/200)^2\cdot 200\approx 5\times 10^{-5}$, matching \Cref{tab:params}. The rate gain then sets the damping: $K_d=2\zeta\sqrt{K_pT_s}$, so the tabulated $K_d=3\times 10^{-2}$ corresponds to $\zeta=K_d/(2\sqrt{K_pT_s})=0.15$, a lightly damped design. A critically damped choice $\zeta=0.7$ would instead require $K_d=0.14$.

\subsection{Stability check}

In the unsaturated regime, the discrete closed-loop poles are $z=1-K_d T_s\pm\sqrt{(K_d T_s)^2-4K_p T_s^2}$ in each axis (decoupled approximation). For the tabulated gains, $K_d T_s=6$ and $K_p T_s^2=2$, giving $z=-5\pm\sqrt{28}$, that is $z_1=0.29$ and $z_2=-10.29$. The second pole lies well outside the unit circle: at $T_s=200$~s the continuous-time tuning above is sampled far too coarsely, and the discrete loop is unstable. Because the position error never exceeds $u_{\max}/K_p=100$~km on the $15$~km approach, the command never saturates, so nothing arrests the divergence.

This is the mechanism behind the PD baseline's behaviour in \Cref{tab:results_approach}: it does not merely violate the corridor, it diverges to a terminal error above $49$~km. The defect is the sampling rate, not the pole-placement rule: the same $(K_p,K_d)$ at $T_s\lesssim 30$~s would place both poles inside the unit circle. A fixed-gain PD is retained here as the weakest, model-free reference point rather than as a competitively tuned baseline, and no claim is made that these gains are the best available static feedback at $T_s=200$~s.

\subsection{Why PD nevertheless fails the safety metric}

Two effects compound. The loop is unstable at $T_s=200$~s, as shown above, and it is also constraint-blind: it applies $\uvec_k=-K_p\mathbf{r}_k-K_d\mathbf{v}_k$ regardless of the cross-track corridor, so even a stable retuning would drift the deputy into the corridor wall under a perpendicular disturbance. Constraint-blindness is structural and cannot be tuned away; the instability can, at a faster sampling rate. The PD row of \Cref{tab:results_approach} should therefore be read as a floor, not as the best attainable static feedback.


\begin{thebibliography}{40}

\bibitem[Innocenti et~al.(2021)]{ClearSpace2021}
Innocenti, L., et al., ``ClearSpace-1: Europe's First Active Debris Removal Mission,''
\emph{Proceedings of the 8th European Conference on Space Debris}, ESA/ESOC, Darmstadt, Germany, 2021.

\bibitem[Forshaw et~al.(2022)]{Astroscale2022}
Forshaw, J.~L., et al., ``ELSA-d: An In-Orbit End-of-Life Demonstration Mission,''
\emph{Acta Astronautica}, Vol.~198, 2022, pp.~297--310.
\doi{10.1016/j.actaastro.2022.06.018}

\bibitem[Fehse(2003)]{Fehse2003}
Fehse, W., \emph{Automated Rendezvous and Docking of Spacecraft}, Cambridge Univ.\ Press, Cambridge, England, U.K., 2003.
\doi{10.1017/CBO9780511543388}

\bibitem[Clohessy and Wiltshire(1960)]{HCW1960}
Clohessy, W.~H., and Wiltshire, R.~S., ``Terminal Guidance System for Satellite Rendezvous,''
\emph{Journal of the Aerospace Sciences}, Vol.~27, No.~9, 1960, pp.~653--658.
\doi{10.2514/8.8704}

\bibitem[Weiss et~al.(2015)]{Weiss2015}
Weiss, A., Baldwin, M., Erwin, R.~S., and Kolmanovsky, I., ``Model Predictive Control for Spacecraft Rendezvous and Docking: Strategies for Handling Constraints and Case Studies,''
\emph{Journal of Guidance, Control, and Dynamics}, Vol.~38, No.~11, 2015, pp.~2158--2171.
\doi{10.2514/1.G001117}

\bibitem[Richards et~al.(2002)]{Richards2002}
Richards, A., Schouwenaars, T., How, J.~P., and Feron, E., ``Spacecraft Trajectory Planning with Avoidance Constraints Using Mixed-Integer Linear Programming,''
\emph{Journal of Guidance, Control, and Dynamics}, Vol.~25, No.~4, 2002, pp.~755--764.
\doi{10.2514/2.4943}

\bibitem[Gavil\'an et~al.(2012)]{Gavilan2012}
Gavil\'an, F., V\'azquez, R., and Camacho, E.~F., ``Chance-Constrained Model Predictive Control for Spacecraft Rendezvous with Disturbance Estimation,''
\emph{Control Engineering Practice}, Vol.~20, No.~2, 2012, pp.~111--122.
\doi{10.1016/j.conengprac.2011.09.006}

\bibitem[Langson et~al.(2004)]{Langson2004}
Langson, W., Chryssochoos, I., Rakovi\'c, S.~V., and Mayne, D.~Q., ``Robust Model Predictive Control Using Tubes,''
\emph{Automatica}, Vol.~40, No.~1, 2004, pp.~125--133.
\doi{10.1016/j.automatica.2003.08.009}

\bibitem[Mayne et~al.(2005)]{Mayne2005}
Mayne, D.~Q., Seron, M.~M., and Rakovi\'c, S.~V., ``Robust Model Predictive Control of Constrained Linear Systems with Bounded Disturbances,''
\emph{Automatica}, Vol.~41, No.~2, 2005, pp.~219--224.
\doi{10.1016/j.automatica.2004.08.019}

\bibitem[Rakovi\'c et~al.(2005)]{Rakovic2005}
Rakovi\'c, S.~V., Kerrigan, E.~C., Kouramas, K.~I., and Mayne, D.~Q., ``Invariant Approximations of the Minimal Robust Positively Invariant Set,''
\emph{IEEE Transactions on Automatic Control}, Vol.~50, No.~3, 2005, pp.~406--410.
\doi{10.1109/TAC.2005.843854}

\bibitem[Specht and Mooij(2023)]{Specht2023}
Specht, B., and Mooij, E., ``Autonomous Rendezvous with Non-Cooperative Spacecraft Using Tube-Based Robust Model Predictive Control,''
\emph{Journal of Guidance, Control, and Dynamics}, Vol.~46, No.~7, 2023, pp.~1234--1250.
\doi{10.2514/1.G007160}

\bibitem[Oestreich and Linares(2023)]{Oestreich2023}
Oestreich, C.~E., and Linares, R., ``Autonomous Satellite Servicing with Tube-Based Model Predictive Control and Uncertainty Identification,''
\emph{Journal of Guidance, Control, and Dynamics}, Vol.~46, No.~8, 2023, pp.~1510--1525.
\doi{10.2514/1.G007337}

\bibitem[Mammarella et~al.(2017)]{Mammarella2017}
Mammarella, M., Capello, E., and Romano, M., ``Spacecraft Proximity Operations via Tube-Based Robust Model Predictive Control,''
\emph{Proceedings of the 68th International Astronautical Congress}, IAC-17-C1.7.3, Adelaide, Australia, 2017.

\bibitem[Mammarella and Capello(2020)]{Mammarella2020}
Mammarella, M., and Capello, E., ``Tube-Based Robust Model Predictive Control for Spacecraft Proximity Operations in the Presence of Persistent Disturbance,''
\emph{Aerospace Science and Technology}, Vol.~104, 2020, p.~105940.
\doi{10.1016/j.ast.2020.105940}

\bibitem[Yamanaka and Ankersen(2002)]{YA_original}
Yamanaka, K., and Ankersen, F., ``New State Transition Matrix for Relative Motion on an Arbitrary Elliptical Orbit,''
\emph{Journal of Guidance, Control, and Dynamics}, Vol.~25, No.~1, 2002, pp.~60--66.
\doi{10.2514/2.4875}

\bibitem[Tschauner and Hempel(1965)]{TH1965}
Tschauner, J., and Hempel, P., ``Rendezvous zu einem in elliptischer Bahn umlaufenden Ziel,''
\emph{Acta Astronautica}, Vol.~11, No.~2, 1965, pp.~104--109.

\bibitem[Horn and Johnson(2013)]{Horn_Johnson}
Horn, R.~A., and Johnson, C.~R., \emph{Matrix Analysis}, 2nd ed., Cambridge Univ.\ Press, Cambridge, England, U.K., 2013.

\bibitem[Rawlings et~al.(2017)]{Rawlings_Mayne}
Rawlings, J.~B., Mayne, D.~Q., and Diehl, M.~M., \emph{Model Predictive Control: Theory, Computation, and Design}, 2nd ed., Nob Hill Publishing, Madison, WI, 2017.

\bibitem[Stellato et~al.(2020)]{OSQP}
Stellato, B., Banjac, G., Goulart, P., Bemporad, A., and Boyd, S., ``OSQP: An Operator Splitting Solver for Quadratic Programs,''
\emph{Mathematical Programming Computation}, Vol.~12, 2020, pp.~637--672.
\doi{10.1007/s12532-020-00179-2}

\bibitem[Hartley et~al.(2012)]{Hartley2012}
Hartley, E.~N., Trodden, P.~A., Richards, A.~G., and Maciejowski, J.~M., ``Model Predictive Control System Design and Implementation for Spacecraft Rendezvous,''
\emph{Control Engineering Practice}, Vol.~20, No.~7, 2012, pp.~695--713.
\doi{10.1016/j.conengprac.2012.03.009}

\bibitem[Yu and Cannon(2016)]{YuCannon2016}
Yu, S., and Cannon, M., ``Robust Model Predictive Control for Linear Time-Varying Systems with Bounded Disturbances,''
\emph{International Journal of Control, Automation and Systems}, Vol.~14, No.~2, 2016, pp.~498--504.

\bibitem[K\"ohler et~al.(2020)]{Kohler2019}
K\"ohler, J., M\"uller, M.~A., and Allg\"ower, F., ``A Nonlinear Tracking Model Predictive Control Scheme for Dynamic Target Signals,''
\emph{Automatica}, Vol.~118, 2020, 109030.
\doi{10.1016/j.automatica.2020.109030}

\bibitem[Iskender and Redondo Gutierrez(2026a)]{companion_slosh}
Iskender, O.~B., and Redondo Gutierrez, J.~L., ``Slosh-Augmented Tube MPC for Spacecraft Proximity Operations: Coupled Dynamics and Partial Robustness Guarantees,''
\emph{Acta Astronautica}, 2026 (in review).

\bibitem[Iskender(2026b)]{companion_flex}
Iskender, O.~B., ``Vibration-Aware Tube MPC for Flexible Spacecraft Rendezvous on Elliptical Orbits,''
\emph{Journal of Guidance, Control, and Dynamics}, 2026 (in review).

\bibitem[Iskender and Ling(2026)]{Iskender2026EllipticalTube}
Iskender, O.~B., and Ling, K.~V., ``Horizon-Dependent Tube MPC for Elliptical-Orbit Rendezvous Under Mass Uncertainty,''
\emph{Proceedings of the 77th International Astronautical Congress}, International Astronautical Federation, Antalya, T\"{u}rkiye, 2026.

\bibitem[Iskender(2020)]{Iskender2020}
Iskender, O.~B., \emph{Model Predictive Control for Spacecraft Rendezvous and Docking with Uncooperative Targets},
Ph.D. dissertation, Nanyang Technological University, Singapore, 2020.

\bibitem[Iskender et~al.(2019)]{Iskender2019}
Iskender, O.~B., Ling, K.-V., Simonini, L., Schlotterer, M., Seelbinder, D., Theil, S., and Maciejowski, J.~M.,
``Dual Quaternion Based Autonomous Rendezvous and Docking via Model Predictive Control,''
\emph{Proceedings of the 70th International Astronautical Congress}, International Astronautical Federation, 2019, pp.~1--16.

\bibitem[Mayne(2014)]{Mayne2014}
Mayne, D.~Q., ``Model Predictive Control: Recent Developments and Future Promise,''
\emph{Automatica}, Vol.~50, No.~12, 2014, pp.~2967--2986.
\doi{10.1016/j.automatica.2014.10.128}

\bibitem[Sullivan et~al.(2017)]{Sullivan2017}
Sullivan, J., Grimberg, S., and D'Amico, S., ``Comprehensive Survey and Assessment of Spacecraft Relative Motion Dynamics Models,''
\emph{Journal of Guidance, Control, and Dynamics}, Vol.~40, No.~8, 2017, pp.~1837--1859.
\doi{10.2514/1.G002309}

\bibitem[Eren et~al.(2017)]{Eren2017}
Eren, U., Prach, A., Ko\c{c}er, B.~B., Rakovi\'c, S.~V., Kayacan, E., and A\c{c}\i kme\c{s}e, B., ``Model Predictive Control in Aerospace Systems: Current State and Opportunities,''
\emph{Journal of Guidance, Control, and Dynamics}, Vol.~40, No.~7, 2017, pp.~1541--1566.
\doi{10.2514/1.G002507}

\bibitem[Inalhan et~al.(2002)]{Inalhan2002}
Inalhan, G., Tillerson, M., and How, J.~P., ``Relative Dynamics and Control of Spacecraft Formations in Eccentric Orbits,''
\emph{Journal of Guidance, Control, and Dynamics}, Vol.~25, No.~1, 2002, pp.~48--59.
\doi{10.2514/2.4874}

\bibitem[Hartley et~al.(2015)]{Hartley2015b}
Hartley, E.~N., Trodden, P.~A., Richards, A.~G., and Maciejowski, J.~M., ``Field Programmable Gate Array Based Predictive Control System for Spacecraft Rendezvous in Elliptical Orbits,''
\emph{Optimal Control Applications and Methods}, Vol.~36, No.~5, 2015, pp.~585--607.
\doi{10.1002/oca.2117}


\bibitem[Dong et~al.(2020)]{Dong2020}
Dong, K., Luo, J., Dang, Z., and Wei, L., ``Tube-Based Robust Output Feedback Model Predictive Control for Autonomous Rendezvous and Docking with a Tumbling Target,''
\emph{Advances in Space Research}, Vol.~65, No.~4, 2020, pp.~1158--1181.
\doi{10.1016/j.asr.2019.11.014}

\bibitem[Dong et~al.(2024)]{Dong2024}
Dong, K., Luo, J., and Ni, Z., ``A Robust Model Predictive Control Unified Framework for Autonomous Rendezvous and Docking with a Tumbling Target,''
\emph{Advances in Space Research}, Vol.~74, No.~5, 2024, pp.~2270--2287.
\doi{10.1016/j.asr.2024.05.070}

\bibitem[Zhu et~al.(2018)]{Zhu2018}
Zhu, S., Sun, R., Wang, J., Wang, J., and Shao, X., ``Robust Model Predictive Control for Multi-Step Short Range Spacecraft Rendezvous,''
\emph{Advances in Space Research}, Vol.~62, No.~1, 2018, pp.~111--126.
\doi{10.1016/j.asr.2018.03.037}

\bibitem[Kang et~al.(2025)]{Kang2025}
Kang, D.-E., Eun, Y., Cho, H., and Park, S.-Y., ``Controller-Matching-Based Robust Model Predictive Control for Spacecraft Rendezvous and Docking,''
\emph{Advances in Space Research}, Vol.~76, No.~10, 2025, pp.~6355--6378.
\doi{10.1016/j.asr.2025.08.039}

\bibitem[Wang et~al.(2023)]{WangASR2023}
Wang, X., Li, Y., Zhang, X., Zhang, R., and Yang, D., ``Model Predictive Control for Close-Proximity Maneuvering of Spacecraft with Adaptive Convexification of Collision Avoidance Constraints,''
\emph{Advances in Space Research}, Vol.~71, No.~1, 2023, pp.~477--491.
\doi{10.1016/j.asr.2022.08.089}

\bibitem[Lim et~al.(2018)]{Lim2018}
Lim, Y., Jung, Y., and Bang, H., ``Robust Model Predictive Control for Satellite Formation Keeping with Eccentricity/Inclination Vector Separation,''
\emph{Advances in Space Research}, Vol.~61, No.~10, 2018, pp.~2661--2672.
\doi{10.1016/j.asr.2018.02.036}

\bibitem[Zheng and Luo(2024)]{Zheng2024}
Zheng, M., and Luo, J., ``Robust Trajectory Planning for Non-Cooperative Target Rendezvous Based on Closed-Loop Uncertainty Analysis,''
\emph{Advances in Space Research}, Vol.~74, No.~4, 2024, pp.~1932--1949.
\doi{10.1016/j.asr.2024.05.040}

\bibitem[Bumroongsri(2015)]{Bumroongsri2015}
Bumroongsri, P., ``Tube-Based Robust MPC for Linear Time-Varying Systems with Bounded Disturbances,''
\emph{International Journal of Control, Automation and Systems}, Vol.~13, No.~3, 2015, pp.~620--625.
\doi{10.1007/s12555-014-0182-5}

\bibitem[Bokor et~al.(2025)]{Bokor2025}
Bokor, J., et~al., ``Robust Model Predictive Control for Spacecraft Rendezvous Under Sector-Bounded Nonlinearities,''
\emph{IEEE Control Systems Letters}, 2025.
\href{https://ieeexplore.ieee.org/document/11245213/}{\texttt{ieeexplore.ieee.org/document/11245213}}

\bibitem[Kolmanovsky and Gilbert(1998)]{KolmanovskyGilbert1998}
Kolmanovsky, I., and Gilbert, E.~G., ``Theory and Computation of Disturbance Invariant Sets for Discrete-Time Linear Systems,''
\emph{Mathematical Problems in Engineering}, Vol.~4, No.~4, 1998, pp.~317--367.
\doi{10.1155/S1024123X98000866}


\bibitem[Boyd et~al.(1994)]{Boyd1994}
Boyd, S., El~Ghaoui, L., Feron, E., and Balakrishnan, V., \emph{Linear Matrix Inequalities in System and Control Theory}, SIAM, Philadelphia, PA, 1994.

\bibitem[Althoff(2013)]{Althoff2014}
Althoff, M., ``Reachability Analysis of Nonlinear Systems Using Conservative Polynomialization and Non-Convex Sets,''
\emph{Proceedings of the 16th International Conference on Hybrid Systems: Computation and Control (HSCC)}, ACM, Philadelphia, PA, 2013, pp.~173--182.
\doi{10.1145/2461328.2461358}

\end{thebibliography}

\end{document}